\documentclass{article}
\usepackage[authoryear,sort]{natbib}
\usepackage{fullpage}
\usepackage{comment}

\usepackage{xspace}

\usepackage[sc]{mathpazo} 
\usepackage[T1]{fontenc}  
\usepackage{microtype}    
\usepackage{abstract}     

\usepackage{amsmath,amssymb,mathtools}
\usepackage{bm}
\usepackage{amsthm}
\usepackage{graphicx,subfig,color,wrapfig,epsfig,epstopdf}
\usepackage{enumitem}
\usepackage{booktabs}
\usepackage{algorithm,algorithmic}

\usepackage[utf8]{inputenc} 
\usepackage[T1]{fontenc}    
\usepackage{hyperref}       
\usepackage{cleveref}
\usepackage{url}            
\usepackage{booktabs}       
\usepackage{amsfonts}       
\usepackage{nicefrac}       
\usepackage{microtype}      
\usepackage[normalem]{ulem}

\def\oX{\overline{X}}

\def\oG{\overline{G}}
\def\x{\bm{x}}
\def\ouralg{D$^2$\xspace}

\newcommand\numberthis{\addtocounter{equation}{1}\tag{\theequation}}

\allowdisplaybreaks[4]

\newcounter{ass_counter}
\newcounter{thm_counter}

\newtheorem{theorem}[thm_counter]{Theorem}
\newtheorem{lemma}[thm_counter]{Lemma}
\newtheorem{corollary}[thm_counter]{Corollary}
\newtheorem{assumption}[ass_counter]{Assumption}

\Crefname{assumption}{Assumption}{Assumptions}

\def\JL{\color{black}}
\newcommand{\xr}[1]{\textcolor{black}{#1}}

\usepackage{authblk}
\author[1]{Hanlin Tang\thanks{\texttt{htang14@ur.rochester.edu}}}
\author[1]{Xiangru Lian\thanks{\texttt{xiangru@yandex.com}}}
\author[4]{Ming Yan\thanks{\texttt{yanm@math.msu.edu}}}
\author[2]{Ce Zhang\thanks{\texttt{ce.zhang@inf.ethz.ch}}}
\author[3,1]{Ji Liu\thanks{\texttt{ji.liu.uwisc@gmail.com}}}
\affil[1]{Department of Computer Science, University of Rochester}
\affil[2]{Department of Computer Science, ETH Zurich}
\affil[3]{Tencent AI Lab}
\affil[4]{Department of Computational Mathematics, Science and Engineering, Michigan
State University}

\begin{document}

\title{D$^2$: Decentralized Training over Decentralized Data}

\maketitle

\begin{abstract}

  While training a machine learning  model using multiple workers, each of
    which collects data from their own data sources, it would be most useful
    when the data collected from different workers can be {\em unique} and {\em
      different}. Ironically, recent analysis of decentralized parallel
  stochastic gradient descent (D-PSGD) relies on the assumption that the data
  hosted on different workers are {\em not too different}. In this paper, we ask
  the question: {\em Can we design a decentralized parallel stochastic gradient
    descent algorithm that is less sensitive to the data variance across
    workers?}

  In this paper, we present D$^2$, a novel decentralized parallel stochastic
  gradient descent algorithm designed for large data variance \xr{among workers}
  (imprecisely, ``decentralized'' data). The core of D$^2$ is a variance
  reduction extension of the standard D-PSGD algorithm, which improves the
  convergence rate from $O\left({\sigma \over \sqrt{nT}} +
    {(n\zeta^2)^{\frac{1}{3}} \over T^{2/3}}\right)$ to $O\left({\sigma \over
      \sqrt{nT}}\right)$ where $\zeta^{2}$ denotes the variance among data on
    different workers. As a result, D$^2$ is robust to data variance among
  workers. We empirically evaluated D$^2$ on image classification
    tasks where each worker has access to only the data of a
    limited set of labels, and find that D$^2$ significantly outperforms
  D-PSGD.

\end{abstract}

\section{Introduction}

Training machine learning models in a decentralized way has attracted intensive
interests recently~\cite{Lian_dsgd,Yuan_dsgd,colin2016gossip}. In the
decentralized setting, there is a set of workers, each of which
collect{\color{black}s } data from different data sources. Instead of sending all
of their data to a centralized place, \xr{these workers only communicate with
  their {\em neighbors}}. \xr{The goal is to get a model that is the same as if
  all data are collected in a centralized place.} Decentralized learning
algorithm is important in scenarios in which centralized communication is
expensive or not possible, or the underlying communication network has high
latency.

For decentralized learning to provide benefit, each user should provides data
that is somehow {\em unique}, i.e., \xr{the variance of data collected from
  different workers are large}. \xr{However, many recent theoretical
  results~\cite{Lian_dsgd,lian2017asynchronous,Nedic_D,Yuan_dsgd} all assume a
  bounded data variance across workers} --- when data hosted on different
workers are very different, these approach could converge slowly, both
empirically and theoretically. In this paper, we aim at bringing this
discrepancy between the current theoretical understanding and the requirements
from {\em some} practical scenarios.

\xr{In this paper, we present D$^2$, a novel decentralized learning algorithm
  designed to be robust under high data variance.} \xr{The structure and
  technique of D$^2$ is built upon standard decentralized parallel stochastic
  gradient descent (D-PSGD), but benefits from an additional variance reduction
  component.} \xr{In the D$^2$ algorithm, each worker stores the stochastic
  gradient and {\JL its} local model in last iterate and linearly combines them
  with the current stochastic gradient and local model.} It results in an
improved convergence rate over D-PSGD by eliminating the data variation among
workers. In particular, the convergence rate is improved from {\color{black} $O\left({\sigma \over \sqrt{nT}} + {(n\zeta^2)^{\frac{1}{3}} \over
      T^{2/3}}\right)$}
 to
$O\left({\sigma \over \sqrt{nT}}\right)$ where \xr{$\zeta^{2}$ is the data variation
  among all workers, $\sigma^{2}$ is the data variance within each worker}, $n$ is
the number of workers, and $T$ is the number of iterations. \xr{We empirically
  show $D^2$ can significantly outperform D-PSGD by training an image
  classification model where each worker has access to only the data of a
  limited set of labels.}


Throughout this paper, we consider the following decentralized optimization:
\begin{equation}
  \min_{\bm{x}\in\mathbb{R}^{N}}\quad f(\bm{x}) {\color{black}:=} {1\over n} \sum_{i=1}^n \overbrace{\mathbb{E}_{\xi\sim\mathcal{D}_i}F_{i}(\bm{x}; \xi)}^{=: f_i(\bm{x})},\label{eq:main}
\end{equation}
where $n$ is the number of workers and $\mathcal{D}_i$ is the local data
distribution for worker $i$. All workers are connected to form a connected
graph. \xr{Each worker can only exchange information with its neighbors.}


\paragraph{Definitions and notations}
Throughout this paper, we use following notations and definitions:
\begin{itemize}
\item $\|\cdot\|_F$ denotes the Frobenius norm of {\color{black} matrices}.
\item $\|\cdot\|$ denotes the $\ell_2$ norm for vectors and the spectral norm
  for matrices.
\item $\nabla f(\cdot)$ denotes the gradient of a function $f$.
\item $f^{*}$ denotes the optimal solution of \eqref{eq:main}.
\item $\lambda_{i}(\cdot)$ denotes the $i$th largest eigenvalue of a matrix.
\item $\bm{x}^{(i)}$ denotes the local model of worker $i$. {\color{black} \item
    $\nabla F_i(\bm{x}^{(i)};\xi^{(i)})$ \xr{denotes a local stochastic gradient
      of worker $i$. }}
\item $\bm{1}=[1,1,\cdots,1]^{\top}\in\mathbb{R}^n$ denotes the all-one vector.
\item In order to organize the algorithm more clearly, here we define the
  concatenation of all local variables, stochastic gradients, and their average
  respectively:
\begin{align*}
X:= & [\bm{x}^{(1)}, \dots ,\bm{x}^{(n)}]\in\mathbb{R}^{N\times n},\\
\overline{X}:= & X\frac{\bm{1}}{n}=\frac{1}{n}\sum_{i=1}^n\bm{x}^{(i)},\\
G(X;\xi)
 :=& [\nabla F_1(\bm{x}^{(1)};\xi^{(1)}) ,\dots ,\nabla F_n(\bm{x}^{(n)};\xi^{(n)})] \in \mathbb{R}^{N\times n},\\
\overline{G}(X,\xi):= & G(X,\xi)\frac{\bm{1}}{n}=\frac{1}{n}\sum_{i=1}^n\nabla F_i(\bm{x}^{(i)};\xi^{(i)}),\\
\nabla f(\overline{X}) := & \sum_{i=1}^{n}\frac{1}{n}\nabla f_i\left({\color{black} \overline{X}}\right),\\
\overline{\nabla f}(X):= & \frac{1}{n}\sum_{i=1}^n\nabla f_i(x^{(i)}),
\end{align*}
{\JL where $\xi$ is the collection of randomly sampled data from all workers}
\end{itemize}

{\JL
\paragraph{Organization}
This paper is organized as follows: Section~\ref{sec:rw} reviews related work about the proposed approach; Section~\ref{sec:prelim} introduces the state-of-the-art decentralized stochastic gradient descent method and its convergence rate; Section~\ref{sec:algorithm} introduces the proposed algorithm and its intuition why it can improves the state-of-the-art approach; and Section~\ref{sec:thm}; Section~\ref{sec:exp} validates the proposed approaches via empirical study; and Section~\ref{sec:conclusion} concludes this paper.
}

\section{Related work} \label{sec:rw}

\xr{In this section, we review the stochastic gradient descent algorithm and its
  decentralized variants, decentralized algorithms, and previous variance
  reduction technologies in this section.}

\paragraph{Stochastic gradient descent (SGD)} The SGD approahces
\citep{Ghadimi_dsgd,Moulines_dsgd,Nemi_dsgd} is quite powerful for solving
large-scale machine learning problems. \xr{It achieves a convergence rate of $O
  \left({1}/\sqrt{T}\right)$}. \xr{As an implementation of SGD, the
  \textsl{Centralized Parallel Stochastic Gradient Descent} (\textbf{C-PSGD}),
  has been widely used in parallel computation.} In C-PSGD, a central worker,
whose job is to perform the \xr{variable updates}, is connected to many leaf
workers that \xr{are} used to compute \xr{stochastic gradients} in parallel.
C-PSGD \xr{has been} applied {\color{black} to many deep learning
  \xr{frameworks}}, such as such as CNTK \citep{Seide:2016:CMO:2939672.2945397},
MXNet \citep{chen2015mxnet}, and \xr{TensorFlow} \citep{abadi2016tensorflow}.
The convergence rate of C-PSGD is $O\left(\frac{1}{\sqrt{nT}}\right)$, which
\xr{shows it can achieve} linear speedup with regards to the number of leaf workers.

\paragraph{Decentralized algorithms} \xr{Centralized algorithms  requires a
  central server to communicate with all other workers~\citep{pmlr-v70-suresh17a}}. In contrast,
decentralized algorithms \xr{can work on any connected network and only} rely on
\xr{the} information exchange between neighbor workers~\citep{kashyap2007quantized,lavaei2012quantized,nedic2009quan}.

\xr{Decentralized algorithms are especially useful under a network with limited
  bandwidth or high latency. {\JL It is more favorable when data privacy is sensitive.}
  These advantages have led to successful applications.} The decentralized
approach for multi-task reinforcement learning \xr{was} studied
in~\citet{omidshafiei2017deep,mhamdi2017dynamic}. In~\citet{colin2016gossip}, a
dual based decentralized algorithm \xr{was} proposed to solve the \xr{pairwise
  function optimization. \citet{shi2014linear} and
  \citet{mokhtari2015decentralized} analyzed the decentralized version of the
  ADMM optimization algorithm.
  An information theoretic approach was used to analyze decentralization in
  \citet{dobbe2017fully}. The decentralized version of (sub-)gradient descent
  was studied in \citet{Nedic_D,Yuan_dsgd}. Its $O(1/\sqrt{T})$ convergence
  requires a diminishing stepsize or a constant stepsize that depends on the
  total number of iterations. This phenomenon happens because of the variance
  between the data in different workers, which we call ``outer variance'' to
  differentiate it from the variance in SGD. Recently, there are several
  deterministic decentralized optimization algorithms that allows a constant
  stepsize. For
  example, EXTRA~\cite{shi2015extra} is the first modification of
  decentralized gradient descent that converges under a constant stepsize. Later this algorithm is extended for problems with the sum of smooth and nonsmooth functions at each node~\cite{shi2015proximal}. However, the stepsize depends on both the Lipschitz constant of the differentiable function and the network structure. NIDS is the first algorithm that has a constant network independent stepsize~\cite{li2017decentralized}. This algorithm was simultaneously proposed by~\citet{yuan2017exact} for the smooth case only using a different approach. For directed networks, the algorithm DIGing is proposed in~\citet{nedic2017network}, where two exchanges are needed in each iteration}. \footnote{\JL To Prof. Yan: could you write couple of sentences to summarize these papers.}


\paragraph{Decentralized parallel stochastic gradient descent (D-PSGD)}
\xr{The D-PSGD algorithm \citep{Nedic_D,ram2010asynchronous,ram2010distributed} requires each
  worker to compute a stochastic gradient and exchange its local model with
  neighbors.} In \citet{duchi2012dual}, a dual averaging based method is proposed for solving the constrained decentralized SGD optimization. \xr{In~\citet{Yuan_dsgd}, the convergence rate for D-PSGD was
  analyzed when the gradient is assumed to be bounded. In~\citet{Lan_dsgd}, a
decentralized primal-dual type method was proposed with a computational
complexity of $O \left({n}/{\epsilon^2}\right)$ for general convex objectives.
\citet{Lian_dsgd} proved that D-PSGD can {\color{black} admits linear speedup
  w.r.t. number of workers with a similar convergence rate like C-PSGD}.} 


\paragraph{Variance reduction technology} There have been many methods developed
for reducing the variance in SGD, including SVRG
\citep{johnson2013accelerating}, SAGA \citep{defazio2014saga}, SAG
\citep{schmidt2017minimizing}, MISO \citep{mairal2015incremental}, and mS2GD
\citep{konevcny2016mini}. However, most of these technologies are just designed
for the centralized approaches. The DSA algorithm \citep{mokhtari2016dsa}
applies the variance reduction \xr{similar to SAGA on strongly convex
  decentralized optimization problems and proved a linear convergence rate.}
However, the speedup property is unclear and a table of all stochastic gradients
need to be stoblack.

\section{Preliminary: decentralized stochastic gradient descent}
\label{sec:prelim}

\xr{The decentralized stochastic gradient
descent~\citep{Lian_dsgd,zhang2017projection,shahrampour2017distributed} allows
each worker (say worker $i$) maintaining its own local variable $\x^{(i)}$.}
During \xr{each} iteration {\color{black} (say, iteration $t$)}, each worker
perform{\color{black}s the} following steps:
\begin{enumerate}
\item \xr{Query its neighbors' local variables.}
\item \xr{Take weighted average with its local variable and neighbors' local
  variables:}
\[
\x^{(i)}_{t+{1\over 2}} = \sum_{j=1}^nW_{{\color{black} ij}}\x^{(j)}_t
\]
where $W_{ij}$ is the {\color{black} $(i,j)$ element of the \xr{matrix $W$},
  $W_{ij} = 0$ means worker $i$ and worker $j$ are not connected}.
\item Perform one stochastic gradient descent step
\[
\x^{(i)}_{t+1} = \x^{(i)}_{t+{1\over 2}} - \gamma \nabla F(\x^{(i)}_t; \xi^{(i)}_t)
\]
where $\xi^{(i)}_t$ \xr{represents the data sampled in worker $i$ at the
  iteration $t$ {\JL following the distribution $\mathcal{D}_i$}.}
\end{enumerate}
From \xr{a} global \xr{point of view}, the \xr{update rule} D-PSGD algorithm can
{\color{black} be} \xr{viewed as}
\[
X_{t+1} = X_{t}W - \gamma G(X_t; \xi_t).
\]
It admits the following rate shown in Theorem \ref{thm:zflasjdkf}.
{
  \begin{theorem} [{Convergence rate of D-PSGD}
    \citep{Lian_dsgd}] \label{thm:zflasjdkf} \xr{Under certain assumptions},
    {\color{black} the output of D-PSGD admits the following inequality}
\begin{align*}
     \frac{1 - \gamma L}{2T}  \sum_{t = 0}^{T -
      1} \mathbb{E} \left\|\overline{\nabla f} (X_t)\right\|^2 + \frac{D_{1}}{T} \sum_{t =
      0}^{T - 1} \mathbb{E} \left\| \nabla f \left( \overline{X}_t \right) \right\|^2
    \leq  \frac{f (0) - f^{\ast}}{\gamma T} + \frac{\gamma L}{2n}\sigma^{2} +   \frac{ \gamma^2L^{2} n \sigma^2 }{(1 - \lambda) D_{2}}
                + \frac{9 \gamma^2L^{2} n  \varsigma^2 }{(1 -\sqrt{ \lambda})^{2}  D_{2} } ,
\end{align*}
where $\rho$ reflects the property of the network, $D_1$ and $D_2$ are defined to be
\begin{align*}
D_1 := &\left( \frac{1}{2} - \frac{9\gamma^{2} L^2 n}{(1 -\sqrt{ \rho})^{2} D_2} \right)\\
D_2 := &\left( 1 - \frac{18\gamma^{2}}{(1 -\sqrt{ \rho})^{2}} nL^2 \right)
\end{align*}
and $\sigma$ and $\varsigma$ measure the variation within each worker and among all workers respectively
\begin{align}
\mathbb{E}_{\xi\sim \mathcal{D}_i} \left\| \nabla F_i (\bm{x}; \xi) - \nabla f_i (\bm{x})\right\|^2 \leqslant & \sigma^2, \quad \forall i, \forall \bm{x},\\
     {1\over n}\sum_{i=1}^n\left\| \nabla f_i (\bm{x})-\nabla f (\bm{x})\right\|^2 \leqslant & \zeta^2 , \quad \forall i, \forall \bm{x}.
\end{align}
\end{theorem}
\xr{Choosing the {\color{black} optimal} steplength $\gamma=\frac{1}{L +
  \sigma\sqrt{\frac{K}{n}} + {\color{black}n^{\frac{1}{3}}}\zeta^{\frac{2}{3}}T^{\frac{1}{3}}}$ we have the
following convergence rate:
}
}
\[
{1\over T}\sum_{t=1}^T\mathbb{E}(\|\nabla f(\overline{X}_t)\|^2) \leq {\color{black} O}\left({\sigma \over \sqrt{nT}} + \frac{{\color{black}n^{\frac{1}{3}}}\zeta^{\frac{2}{3}}}{T^{\frac{2}{3}}} {\color{black} + \frac{1}{T}}\right).
\]

The proposed \ouralg algorithm can improve the convergence rate by
{{\color{black}removing the dependence to the global bound of outer variance
    $\zeta$.}

\section{The \xr{$D^{2}$} algorithm}
\label{sec:algorithm}

\begin{algorithm}[h]
\caption{\ouralg algorithm}\label{alg2}
\begin{minipage}{1.0\linewidth}
\begin{algorithmic}[1]
  \STATE {\bfseries Input:} Initial point $\bm{x}^{(i)}_0=\bm{0}$, iteration
  step length $\gamma$, confusion matrix $W$, and \xr{the total number of
    iterations $T$}  \FOR{t = 0,1,2,...,T}

  \STATE Randomly sample $\xi^{(i)}_t$ from \xr{the local data} of the
  {\color{black}$i$th} worker{\color{black}.}
  \STATE Compute a local stochastic gradient based on $\xi^{(i)}_k$ and current
  optimization variable $\bm{x}^{(i)}_t:\nabla
  F_i(\bm{x}^{(i)}_t;\xi^{(i)}_t)${\color{black}.}
\STATE 
\IF{t=0}
\STATE
$\bm{x}_{t+\frac{1}{2}}^{(i)} = \bm{x}_{t}^{(i)} - \gamma\nabla F_i(\bm{x}^{(i)}_t;\xi^{(i)}_t),
$
\ELSE
\STATE 
$\bm{x}_{t+\frac{1}{2}}^{(i)} = 2\bm{x}_{t}^{(i)} - \bm{x}_{t-1}^{(i)}
- \gamma\nabla F_i(\bm{x}^{(i)}_t;\xi^{(i)}_t) + \gamma\nabla F_i(\bm{x}^{(i)}_{t-1};\xi^{(i)}_{t-1}).
$
\ENDIF
\STATE Each worker sends $\bm{x}_{t+\frac{1}{2}}^{(i)}$ to its neighbors, and take the weighted average
\begin{align*}
\bm{x}_{t+1}^{(i)}=\sum_{j=1}^{n}W_{{\color{black} ij}}\bm{x}^{(j)}_{t+\frac{1}{2}},
\end{align*}
where $\bm{x}_{t+\frac{1}{2}}^{(j)}$ is from \xr{the} worker $j$.\ENDFOR
\STATE {\bfseries Output:} $\frac{1}{n}\sum_{i=1}^{n}\bm{x}^{(i)}_T$
\end{algorithmic}
\end{minipage}
\end{algorithm}

{\color{black} In \ouralg algorithm, each worker repeats the following updating
  rule (say, at iteration $t$) for worker $i$}
\begin{enumerate}
\item \xr{ Compute a} local stochastic gradient $\nabla
  F(\bm{x}_t^{(i)};\xi_t^{(i)})$ {\color{black} by sampling $\xi_t^{(i)}$
    \xr{from} distribution $\mathcal{D}^{(i)}$};
\item Update the local model $\bm{x}_{t+\frac{1}{2}}^{(i)} \gets
  2\bm{x}_{t}^{(i)} - \bm{x}_{t-1}^{(i)} - \gamma\nabla
  F_i\left(\bm{x}^{(i)}_t;\xi^{(i)}_t\right) + \gamma\nabla
  F_i\left(\bm{x}^{(i)}_{t-1};\xi^{(i)}_{t-1}\right) $ using \xr{the} local
  models and stochastic gradients \xr{in} \xr{both the} $t$th iteration and
  \xr{the} $(t-1)$th iteration.
\item When the synchronization barrier is met, exchange
  $\bm{x}_{t+\frac{1}{2}}^{(i)}$ \xr{with neighbors:}
 \begin{align*}
\bm{x}_{t+1}^{(i)}=\sum_{j=1}^{n}W_{{\color{black} ij}}\bm{x}^{(j)}_{t+\frac{1}{2}}.
 \end{align*}

\end{enumerate}
\xr{From a global point of view, the update rule of \ouralg can be viewed as:}
\begin{align*}
X_{t+1} = & \left(2X_t - X_{t-1} - \gamma G(X_t; \xi_t) + \gamma G(X_{t-1}; \xi_{t-1})\right)W. 
\end{align*}
{\color{black} The complete algorithm is summarized in Algorithm~\ref{alg2}. }

\paragraph{\ouralg essentially runs the stochastic gradient descent step.}
To understand the intuition of \ouralg, let us consider the mean value
$\overline{X}_t$, {\color{black}which} \xr{gets updated just like the standard
stochastic gradient descent:}
\begin{small}
\begin{align*}
\overline{X}_{t+1} = & \left(2X_t - X_{t-1} - \gamma G(X_t; \xi_t) + \gamma G(X_{t-1}; \xi_{t-1})\right)W\frac{\bm{1}_n}{n},\\
\overline{X}_{t+1} = & 2\overline{X}_t - \overline{X}_{t-1} - \gamma \overline{G}(X_t; \xi_t) + \gamma \overline{G}(X_{t-1}; \xi_{t-1}),
\end{align*}
\end{small}
or equivalently
\begin{align*}
 \overline{X}_{t+1} - \overline{X}_t
 = & \overline{X}_t - \overline{X}_{t-1} - \gamma \overline{G}(X_t; \xi_t) + \gamma \overline{G}(X_{t-1}; \xi_{t-1}),\\
 = & \overline{X}_1 - \overline{X}_{0} - \gamma\sum_{k=1}^t\left(\overline{G}(X_t; \xi_t) - \overline{G}(X_{t-1}; \xi_{t-1})\right)\\
= & -\gamma \overline{G}(X_t; \xi_t).\quad\text{(due to $X_1 = X_0 -\gamma G(X_0; \xi_0)$)}.\numberthis \label{global_mean_evolution}
\end{align*}

\paragraph{Why \ouralg improves the D-PSGD?}
Acute reviewers may notice that the D-PSGD algorithm also essentially updates in
the form of stochastic gradient descent in \eqref{global_mean_evolution}. Then
why D$^2$ can improve D-PSGD?

Assume that {\JL $X_{t}$} has achieved the optimum $X^*:= x^* {\bf 1}^{\top}$ with
all local models equal to the optimum $x^*$ to \eqref{eq:main}. {\color{black}
  Then for D-PSGD}, \xr{the next update will be}
\[
X_{t+1} = X^* - \gamma G(X^*; \xi_{t}).
\]
\xr{It shows that the convergence when we approach a solution is affected by}
$\mathbb{E}[\|G(X^*; \xi_{t}\|_F^2]${\color{black},} {\color{black}which} is
bounded by
\[
\mathcal{O}(\sigma^2 + \zeta^2).
\]
\xr{as we can see from the following:}
\begin{align*}
& \mathbb{E}[\|G(X^*; \xi_{t}\|^2_F]\\
 = & \mathbb{E}\sum_{i=1}^n\left\|\left(\nabla F_i(\bm{x}^{*};\xi^{(i)}_{t+1}) - \nabla f_i(\bm{x}^{*})\right) + \nabla f_i(\bm{x}^{*})\right\|^2\\
\leq & 2\mathbb{E}\sum_{i=1}^n\left\|\left(\nabla F_i(\bm{x}^{*};\xi^{(i)}_{t+1}) - \nabla f_i(\bm{x}^{*})\right)\right\|^2 + 2\left\|\nabla f_i(\bm{x}^{*}) - \nabla f(\bm{x}^{*})\right\|^2\\
\leq & 2\sigma^2 + 2\zeta^2.
\end{align*}
Next we apply a similar analysis for \ouralg by assuming that both $X_{t-1}$ and
$X_t$ have \xr{reached {\color{black} the optimal} solution $X^{*}$}. \xr{The next update for $D^{2}$ will
  be:}
\[
X_{t+1} = \left(X^* - \gamma G(X^*; \xi_t) - \gamma G(X^*; \xi_{t-1})\right)W.
\]
\xr{It shows that for $D^{2}${\color{black}, the convergence when we approach a solution relies on the magnitude of
  $\mathbb{E}[\|G(X^*; \xi_t)- G(X^*; \xi_{t-1}\|^2_F]$, which is bounded by:}
 }
\[
\mathcal{O}(\sigma^2),
\]
\xr{which can be seem from:}
\begin{align*}
 \mathbb{E}[\|G(X^*; \xi_t)- G(X^*; \xi_{t-1})\|^2_F 
= &\mathbb{E}\sum_{i=1}^n\left\|\nabla F_i(\bm{x}^*; \xi_t^{(i)})- \nabla f_i(\bm{x}^*)\right\|^2
 - \mathbb{E}\sum_{i=1}^n\left\|\nabla F_i(\bm{x}^*; \xi_{t-1}^{(i)})- \nabla f_i(\bm{x}^*)\right\|^2
\\ \leq & 2\sigma^2.
\end{align*}
%
%

\section{Theoretical guarantee}
\label{sec:thm}

{This section provides the theoretical guarantee for the proposed \ouralg algorithm. We first give the assumptions requiblack below.}

\begin{assumption}
  \label{ass:global}
  Throughout this paper, we make the following commonly used assumptions:
  \hfill
  \begin{enumerate}
  \item \textbf{Lipschitzian gradient:} All function $f_i(\cdot)$'s are with $L$-Lipschitzian gradients.
      \item \textbf{Bounded variance:} \xr{Assume bounded variance of stochastic gradient {\emph{within} each worker}}
\xr{      \begin{align*}
      \mathbb{E}_{\xi\sim \mathcal{D}_i} \left\| \nabla F_i (\bm{x}; \xi) - \nabla f_i (\bm{x})\right\|^2 \leqslant & \sigma^2, \quad \forall i, \forall \bm{x}.
\end{align*}
} \item \textbf{Symmetric confusion matrix:}
The confusion matrix $W$ is symmetric and satisfies $W\bm{1}=\bm{1}$.
\item \textbf{Spectral gap:} Let the eigenvalues of $W\in \mathbb{R}^{n\times
    n}$ be $\lambda_1 \geq \lambda_2 \geq \cdots \geq \lambda_n$. Denote
  by {\color{black} for short
\begin{align*}
\lambda := \max \limits_{i\in \{2,\cdots, n\}} \lambda_i = \lambda_2.
\end{align*}
We assume $\lambda<1$ and $\lambda_n>-\frac{1}{3}$.}
\item \textbf{Initialization:} \xr{W.l.o.g.,} assume all \xr{local variables}
  are initialized by zero, that is, $X_0 = 0$.
  \label{Ass:init} 
  \end{enumerate}
\end{assumption}

Existing decentralized consensus
algorithms~\citep{shi2015proximal,li2017decentralized} use a modification of the
doubly stochastic matrix such that $\lambda>0$, i.e., choose $W=(\tilde{W}+I)/2$
where $W$ is a doubly stochastic matrix. Recently,~\citet{li2017primal} show
that $\lambda_n>-1/3$ is optimal in the convergence of EXTRA. However, the
optimal $\lambda_n$ for NIDS \citep{li2017decentralized} is {\color{black} unknown}. In this paper,
{\color{black}we proved that $-\frac{1}{3}$ is the infimum of $\lambda_n$,} and
when it blackuces to deterministic case, this condition is weaker than that
in~\cite{li2017decentralized}. This is important, because we actually can use a
$W$ that performs better.

\xr{Given Assumption \ref{ass:global}, we have following convergence guarantee
  for $D^{2}$: }

\begin{theorem}[Convergence of Algorithm~\ref{alg2}] \label{theo_1}
Choose the steplength $\gamma$ in Algorithm~\ref{alg2} to be a constant {satisfying} $1-24C_2\gamma^2L^2>0$. Under Assumption~\ref{ass:global}, we have the following convergence rate for Algorithm~\ref{alg2}:
\begin{align*}
& A_1\|\nabla f(\bm{0})\|^2 + \sum_{t=1}^{T-1}\left(\mathbb{E}\|\nabla f({\color{black}\overline{X}_t})\|^2 + A_2\mathbb{E}\|\overline{\nabla f}(X_t)\|^2\right)\\
\leq & \frac{2(f(0)-f^*)}{\gamma}
 + \frac{LT\gamma}{n}\sigma^2  + \frac{6L^2C_1\gamma^2\zeta_0^2}{C_3} \numberthis \label{eq:thm2} 
  +\frac{12L^2C_2\gamma^2\sigma^2T}{C_3} + \frac{6L^2C_2\gamma^4L^2\sigma^2T}{nC_3}
  + \frac{6L^2C_1\gamma^2\sigma^2}{C_3},
\end{align*}
where
\begin{align*}
\zeta_0 := & \frac{1}{n}\sum_{i=1}^n \|\nabla f_i(\bm{0}) - \nabla f(\bm{0})\|^2,\\
v := & \lambda_n -\sqrt{\lambda_n^2 -\lambda_n},\\
C_1 := & \max \left\{\frac{1}{1-|v|^2},\frac{1}{(1-\lambda)^2}\right\} \geq 1,\\
C_2 := & \max \left\{\frac{\lambda_n^2}{(1-|v|^2)},\frac{\lambda^2}{(1-\sqrt{\lambda})^2(1-\lambda)}\right\},\\
C_3 := & 1-24C_2\gamma^2L^2,\\
A_1 := & 1 - \frac{6L^2C_1\gamma^2}{C_3},\\
A_2 := & 1-L\gamma-\frac{6L^2C_2\gamma^4L^2}{C_3}.
\end{align*}
\end{theorem}
 By appropriately specifying the step length $\gamma$ we
  reach the following corollary:}

\begin{corollary}\label{corollary2}
Choose the {\color{black}step length} $\gamma$ in Algorithm~\ref{alg2} to be $\gamma = \frac{1}{8\sqrt{C_2}L + 6\sqrt{C_1}L + \sigma\sqrt{\frac{T}{n}}} $, {\color{black}where $C_1$ and $C_2$ are defined in Theorem~\ref{theo_1}}. Under Assumption~\ref{ass:global}, the following convergence rate holds
\begin{align*}
\frac{1}{T}\sum_{t=0}^T\mathbb{E}\|\nabla f(\overline{X_t})\|^2
\lesssim & \frac{\sigma}{\sqrt{nT}} + \frac{1}{T}  + \frac{\zeta_0^2}{T + \sigma^2T^2}
 + \frac{\sigma^2}{1+ \sigma^2T},
\end{align*}
where $\zeta_0$ is defined in Theorem~\ref{theo_1} and we treat $f(0)- f^*$, $L$, $\lambda_n$, and $\lambda$ as constants.
\end{corollary}

Note that we can obtain {\JL even better} constants by choosing different parameters and applying tighter inequalities, however, the main result of this corollary is to show the order of the convergence. {\JL We highlight a few key observations from our theoretical results in the following.}
\begin{description}
\item [Tightness of the convergence rate] Setting $\sigma = 0$ and $\zeta_0=0$, which blackuces the VR-SGD to a normal GD algorithm, we shall see that the convergence rate becomes $O\left(\frac{1}{T}\right)$, which is exactly the rate of GD.
\item [Linear speedup] Since the leading term of the convergence rate is $O
  \left(\frac{1}{\sqrt{nT}}\right)$, which is consistent with the convergence
  rate of C-PSGD, this indicates that we would achieve a linear speed up with
  respect to the number of nodes. {\color{black} \item [Consistent with NIDS] In
    NIDS \citep{li2017primal}, the term depends on $\zeta_0$ in the convergence
    rate is $O\left(\frac{\zeta^2_0}{T} \right)$. While the corresponding term
    in $D^2$ is $O\left(\frac{\zeta_0^2}{T+\sigma^2T^2} \right) $, \xr{which
    indicates when our algorithm is consistent with NIDS because in NIDS
    $\sigma$ is consideblack to be 0.}}
\item [Superiority over D-PSGD] When compablack to D-PSGD, the convergence rate of {\color{black}$D^2$} only depends on $\zeta_0$, and the corresponding decaying rate is $\frac{\zeta_0}{T^2}$. Whereas in D-PSGD \citep{Lian_dsgd}, we need to assume an upper bound for the global variance between different nodes' dataset, and its influence can be compablack to $\sigma^2$, the inner variance of each node itself. This means we can always achieve a much better convergence rate than D-PSGD.
\end{description}

\section{Experiments}
\label{sec:exp}

We evaluate the effectiveness of D$^2$ by comparing it with both centralized and
decentralized SGD \xr{algorithms}.

\subsection{Experiment Settings}

We conduct experiments in two settings.

\begin{enumerate}
\vspace{-1em}
\item \textsc{TransferLearning}: We test the case that
each worker has access to a local pre-trained neural network
as feature extractor, and we want to train a logistic regression
model among all these workers. In our experiment, we
select the first 16 classes of ImageNet and use InceptionV4 as
the feature extractor to extract 2048 features for
each image. We conduct data augmentation and generate a
blurblack version for each image. In total this datasaet contains
16$\times$1300$\times$2 images.
\vspace{-0.5em}
\item \textsc{LeNet}: We test the case that all workers \xr{collaboratively}
  train a neural network model. We train a LeNet on the CIFAR10 dataset. In
  total this dataset contains 50,000 images of size 32$\times$32.
\end{enumerate}
\vspace{-3mm}

One caveat of training more recent neural networks is that
modern architectures often have a batch normalization layer,
which inherently assumes that the data distribution is
uniform across different batches, which is not the case that
we are interested in. In principle, we could
also flow the batch information through the network in a
decentralized way; however, we leave this as future work.

By default, each worker only has {\em exclusive} access to
a subset of classes. For \textsc{TransferLearning}, we use
16 workers and each worker has access to one class; for
\textsc{LeNet}, we use 5 workers and each worker has access to
two classes. For comparison, we also consider a case when
the datasets is first shuffled and then uniformly partitioned
among all the workers, we call this the {\em shuffled case},
and the default one the {\em unshuffled case}. We use a
ring topology for both experiments.

{\bf Parameter Tuning.} For \textsc{TransferLearning}, we
use constant learning rates and tune it from \{0.01, 0.025, 0.05, 0.075, 0.1\}.
For \textsc{LeNet}, we use constant learning rate 0.05 which is tuned from \{0.5, 0.1, 0.05, 0.01\} for centralized algorithms and batch size 128 on each worker.

\xr{{\bf Metrics.} In this paper, we mainly focus on the convergence rate of
  different algorithms instead of the wall clock speed. This is because the
  implementation of D$^2$ is a minor change over the standard D-PSGD algorithm,
  and thus they has almost the same speed to finish one epoch of training, and
  both are no slower than the centralized algorithm. When the network has high
  latency, if a decentralized algorithm ($D^{2}$ or D-PSGD) converges with a
  similar speed as the centralized algorithm, it can be up to one order of
  magnitude faster \cite{Lian_dsgd}. However, the convergence rate depending on
  the ``outer variance'' is different for both algorithms. }
\subsection{Unshuffled Case}
\vspace{-3mm} We are mostly interested in the unshuffled case, in which the data
variation across workers is maximized. Figure~\ref{fig:unshuffled} shows the
result. \xr{In the unshuffled case,} we see that the 
D-PSGD algorithm convergences slower than the centralized case. This is
consistent with the original D-PSGD paper \citep{Lian_dsgd}. On the other hand,
D$^2$ converges much faster than D-PSGD, and achieves almost the same loss as
the centralized algorithm. For the LeNet case, each worker only has access to
data of assigned two labels, which means the data variation is very large. The
D-PSGD does not converge with the the given learning rate 0.05.\footnote{We can
  tune the learning rate 50x smaller for D-PSGD to converge in this case, but
  doing so will make D-PSGD stuck at the starting point for quite a long time.}

\begin{figure*}[!htb]
\vspace{-4mm}
\centering
\includegraphics[width=0.85\textwidth]{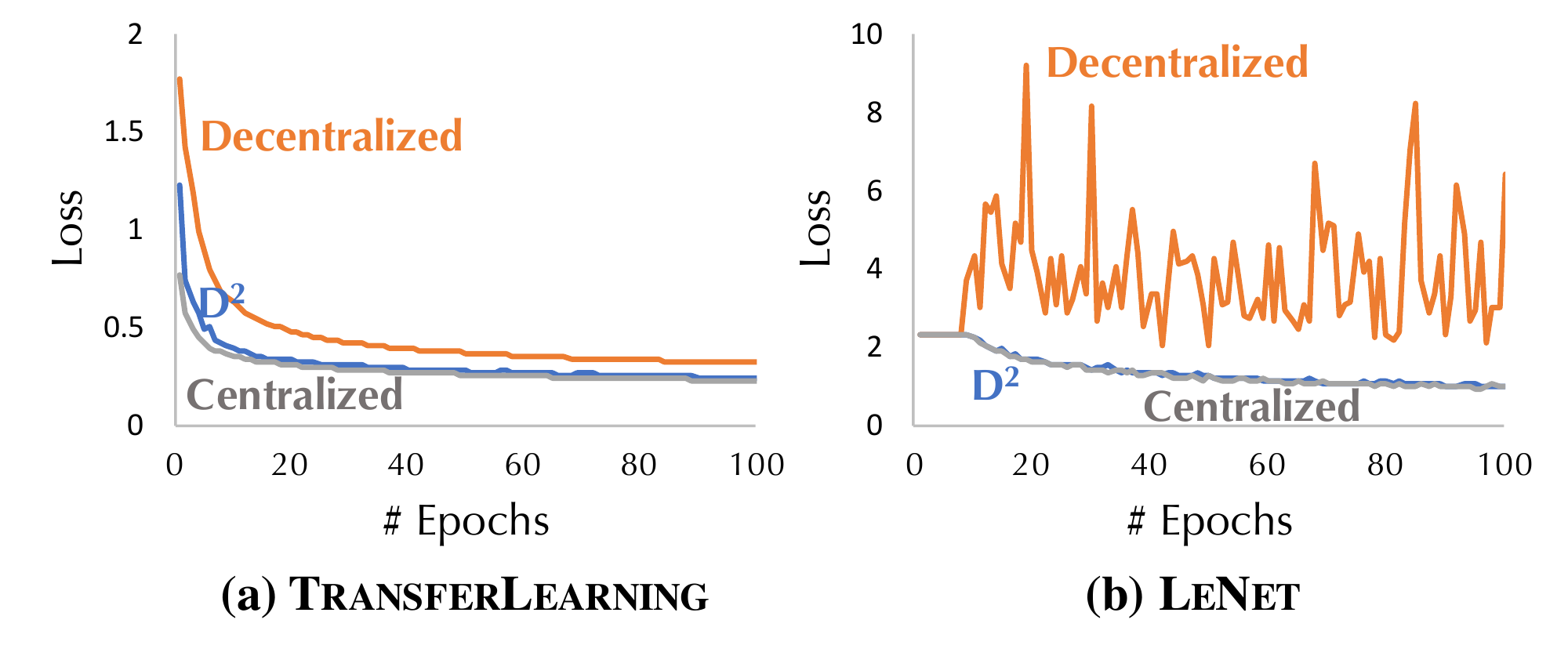}
\vspace{-6mm}
\caption{Convergence of Different Distributed Training Algorithms (Unshuffled Case).}
\label{fig:unshuffled}
\end{figure*}
\begin{figure*}[!htb]
\vspace{-4mm}
\centering
\includegraphics[width=0.85\textwidth]{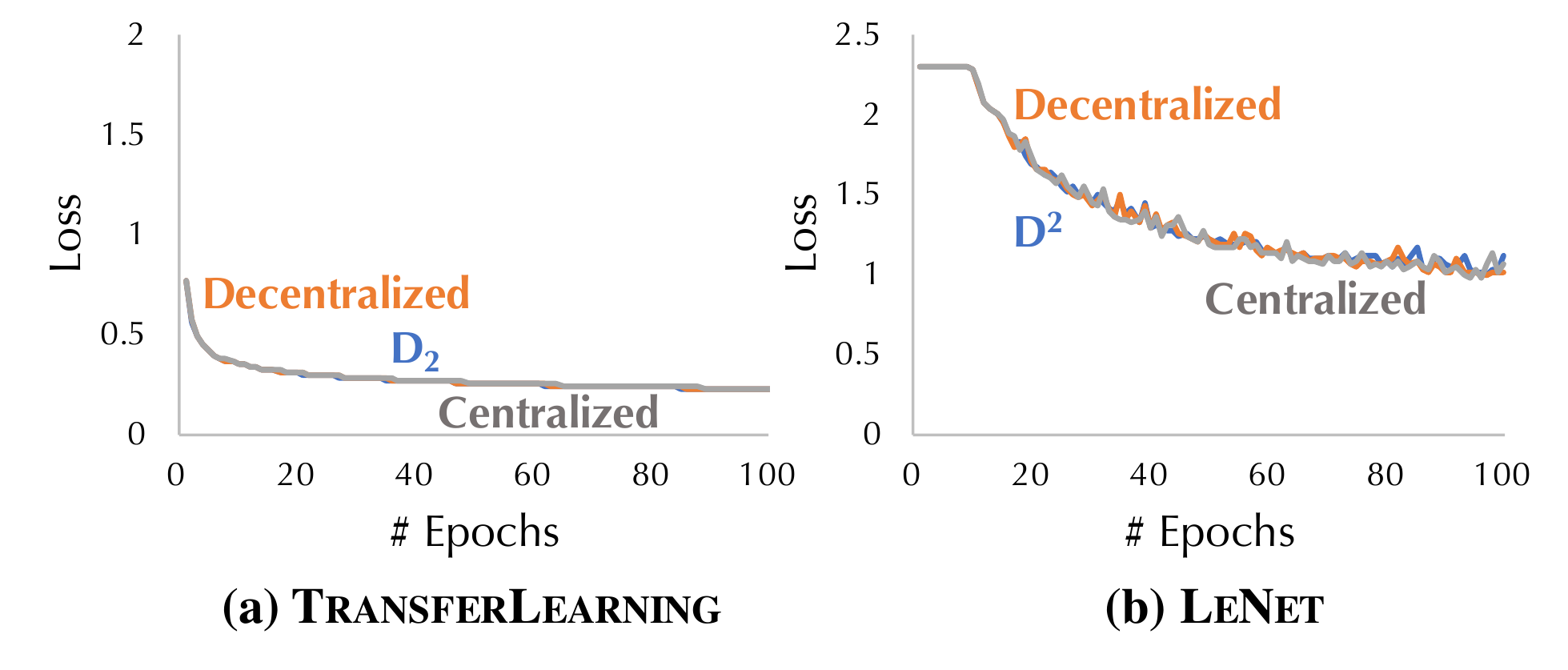}
\vspace{-6mm}
\caption{Convergence of Different Distributed Training Algorithms (Shuffled Case).}
\label{fig:shuffled}
\end{figure*}

\subsection{Shuffled Case}
\vspace{-3mm}

As a sanity check, Figure~\ref{fig:shuffled} shows the result of three different
algorithms on the shuffled data. In this case, the data variation of among
workers is small (in expectation, they are drawn from the same distribution). We
see that, all strategies have similar convergence rate. This validate that the
D$^2$ algorithm is more effective for larger data variation between different
workers.

\section{Conclusion}
\label{sec:conclusion}

In this paper, we propose a decentralized algorithm, namely, D$^2$ algorithm.
D$^2$ algorithm integrates the D-PSGD algorithm \xr{with} the variance reduction
technology, \xr{by which we improves the convergence rate of D-PSGD.} The
variance reduction technology used in this paper is different from the commonly
used ones such as SVRG and SAGA, that are designed for centralized approaches.
Experiments validate the advantage of D$^2$ over D-PSGD --- D$^2$ converges with
a rate that is similar to centralized SGD while D-PSGD does not converge to the
\xr{a solution with a similar quality} when the data variance is large.
\xr{While being robust to large data variance among workers}, the same
performance benefit of D-PSGD over the centralized strategy still holds for
D$^2$.

\bibliographystyle{abbrvnat}
\bibliography{vsgd}

\begin{thebibliography}{38}
\providecommand{\natexlab}[1]{#1}
\providecommand{\url}[1]{\texttt{#1}}
\expandafter\ifx\csname urlstyle\endcsname\relax
  \providecommand{\doi}[1]{doi: #1}\else
  \providecommand{\doi}{doi: \begingroup \urlstyle{rm}\Url}\fi

\bibitem[Abadi et~al.(2016)Abadi, Barham, Chen, Chen, Davis, Dean, Devin,
  Ghemawat, Irving, Isard, et~al.]{abadi2016tensorflow}
M.~Abadi, P.~Barham, J.~Chen, Z.~Chen, A.~Davis, J.~Dean, M.~Devin,
  S.~Ghemawat, G.~Irving, M.~Isard, et~al.
\newblock Tensorflow: A system for large-scale machine learning.
\newblock In \emph{OSDI}, volume~16, pages 265--283, 2016.

\bibitem[Chen et~al.(2015)Chen, Li, Li, Lin, Wang, Wang, Xiao, Xu, Zhang, and
  Zhang]{chen2015mxnet}
T.~Chen, M.~Li, Y.~Li, M.~Lin, N.~Wang, M.~Wang, T.~Xiao, B.~Xu, C.~Zhang, and
  Z.~Zhang.
\newblock Mxnet: A flexible and efficient machine learning library for
  heterogeneous distributed systems.
\newblock \emph{arXiv preprint arXiv:1512.01274}, 2015.

\bibitem[Colin et~al.(2016)Colin, Bellet, Salmon, and
  Cl{\'e}men{\c{c}}on]{colin2016gossip}
I.~Colin, A.~Bellet, J.~Salmon, and S.~Cl{\'e}men{\c{c}}on.
\newblock Gossip dual averaging for decentralized optimization of pairwise
  functions.
\newblock In \emph{International Conference on Machine Learning}, pages
  1388--1396, 2016.

\bibitem[Defazio et~al.(2014)Defazio, Bach, and
  Lacoste-Julien]{defazio2014saga}
A.~Defazio, F.~Bach, and S.~Lacoste-Julien.
\newblock Saga: A fast incremental gradient method with support for
  non-strongly convex composite objectives.
\newblock In \emph{Advances in neural information processing systems}, pages
  1646--1654, 2014.

\bibitem[Dobbe et~al.(2017)Dobbe, Fridovich-Keil, and Tomlin]{dobbe2017fully}
R.~Dobbe, D.~Fridovich-Keil, and C.~Tomlin.
\newblock Fully decentralized policies for multi-agent systems: An information
  theoretic approach.
\newblock In \emph{Advances in Neural Information Processing Systems}, pages
  2945--2954, 2017.

\bibitem[Duchi et~al.(2012)Duchi, Agarwal, and Wainwright]{duchi2012dual}
J.~C. Duchi, A.~Agarwal, and M.~J. Wainwright.
\newblock Dual averaging for distributed optimization: Convergence analysis and
  network scaling.
\newblock \emph{IEEE Transactions on Automatic control}, 57\penalty0
  (3):\penalty0 592--606, 2012.

\bibitem[Ghadimi and Lan(2013)]{Ghadimi_dsgd}
S.~Ghadimi and G.~Lan.
\newblock Stochastic first- and zeroth-order methods for nonconvex stochastic
  programming.
\newblock \emph{SIAM Journal on Optimization}, 23\penalty0 (4):\penalty0
  2341--2368, 2013.
\newblock \doi{10.1137/120880811}.

\bibitem[Johnson and Zhang(2013)]{johnson2013accelerating}
R.~Johnson and T.~Zhang.
\newblock Accelerating stochastic gradient descent using predictive variance
  reduction.
\newblock In \emph{Advances in neural information processing systems}, pages
  315--323, 2013.

\bibitem[Kashyap et~al.(2007)Kashyap, Ba{\c{s}}ar, and
  Srikant]{kashyap2007quantized}
A.~Kashyap, T.~Ba{\c{s}}ar, and R.~Srikant.
\newblock Quantized consensus.
\newblock \emph{Automatica}, 43\penalty0 (7):\penalty0 1192--1203, 2007.

\bibitem[Kone{\v{c}}n{\`y} et~al.(2016)Kone{\v{c}}n{\`y}, Liu, Richt{\'a}rik,
  and Tak{\'a}{\v{c}}]{konevcny2016mini}
J.~Kone{\v{c}}n{\`y}, J.~Liu, P.~Richt{\'a}rik, and M.~Tak{\'a}{\v{c}}.
\newblock Mini-batch semi-stochastic gradient descent in the proximal setting.
\newblock \emph{IEEE Journal of Selected Topics in Signal Processing},
  10\penalty0 (2):\penalty0 242--255, 2016.

\bibitem[Lan et~al.(2017)Lan, Lee, and Zhou]{Lan_dsgd}
G.~Lan, S.~Lee, and Y.~Zhou.
\newblock Communication-efficient algorithms for decentralized and stochastic
  optimization.
\newblock 01 2017.

\bibitem[Lavaei and Murray(2012)]{lavaei2012quantized}
J.~Lavaei and R.~M. Murray.
\newblock Quantized consensus by means of gossip algorithm.
\newblock \emph{IEEE Transactions on Automatic Control}, 57\penalty0
  (1):\penalty0 19--32, 2012.

\bibitem[Li and Yan(2017)]{li2017primal}
Z.~Li and M.~Yan.
\newblock A primal-dual algorithm with optimal stepsizes and its application in
  decentralized consensus optimization.
\newblock \emph{arXiv preprint arXiv:1711.06785}, 2017.

\bibitem[Li et~al.(2017)Li, Shi, and Yan]{li2017decentralized}
Z.~Li, W.~Shi, and M.~Yan.
\newblock A decentralized proximal-gradient method with network independent
  step-sizes and separated convergence rates.
\newblock \emph{arXiv preprint arXiv:1704.07807}, 2017.

\bibitem[Lian et~al.(2017{\natexlab{a}})Lian, Zhang, Zhang, Hsieh, Zhang, and
  Liu]{Lian_dsgd}
X.~Lian, C.~Zhang, H.~Zhang, C.-J. Hsieh, W.~Zhang, and J.~Liu.
\newblock Can decentralized algorithms outperform centralized algorithms? a
  case study for decentralized parallel stochastic gradient descent.
\newblock 05 2017{\natexlab{a}}.

\bibitem[Lian et~al.(2017{\natexlab{b}})Lian, Zhang, Zhang, and
  Liu]{lian2017asynchronous}
X.~Lian, W.~Zhang, C.~Zhang, and J.~Liu.
\newblock Asynchronous decentralized parallel stochastic gradient descent.
\newblock \emph{arXiv preprint arXiv:1710.06952}, 2017{\natexlab{b}}.

\bibitem[Mairal(2015)]{mairal2015incremental}
J.~Mairal.
\newblock Incremental majorization-minimization optimization with application
  to large-scale machine learning.
\newblock \emph{SIAM Journal on Optimization}, 25\penalty0 (2):\penalty0
  829--855, 2015.

\bibitem[Mhamdi et~al.(2017)Mhamdi, Mahdi, Hendrikx, Guerraoui, and
  Maurer]{mhamdi2017dynamic}
E.~Mhamdi, E.~Mahdi, H.~Hendrikx, R.~Guerraoui, and A.~D.~O. Maurer.
\newblock Dynamic safe interruptibility for decentralized multi-agent
  reinforcement learning.
\newblock Technical report, EPFL, 2017.

\bibitem[Mokhtari and Ribeiro(2015)]{mokhtari2015decentralized}
A.~Mokhtari and A.~Ribeiro.
\newblock Decentralized double stochastic averaging gradient.
\newblock In \emph{Signals, Systems and Computers, 2015 49th Asilomar
  Conference on}, pages 406--410. IEEE, 2015.

\bibitem[Mokhtari and Ribeiro(2016)]{mokhtari2016dsa}
A.~Mokhtari and A.~Ribeiro.
\newblock Dsa: Decentralized double stochastic averaging gradient algorithm.
\newblock \emph{Journal of Machine Learning Research}, 17\penalty0
  (61):\penalty0 1--35, 2016.

\bibitem[Moulines and Bach(2011)]{Moulines_dsgd}
E.~Moulines and F.~R. Bach.
\newblock Non-asymptotic analysis of stochastic approximation algorithms for
  machine learning.
\newblock In J.~Shawe-Taylor, R.~S. Zemel, P.~L. Bartlett, F.~Pereira, and
  K.~Q. Weinberger, editors, \emph{Advances in Neural Information Processing
  Systems 24}, pages 451--459. Curran Associates, Inc., 2011.

\bibitem[Nedic and Ozdaglar(2009)]{Nedic_D}
A.~Nedic and A.~Ozdaglar.
\newblock Distributed subgradient methods for multi-agent optimization.
\newblock \emph{IEEE Transactions on Automatic Control}, 54\penalty0
  (1):\penalty0 48--61, 2009.

\bibitem[Nedic et~al.(2009)Nedic, Olshevsky, Ozdaglar, and
  Tsitsiklis]{nedic2009quan}
A.~Nedic, A.~Olshevsky, A.~Ozdaglar, and J.~N. Tsitsiklis.
\newblock On distributed averaging algorithms and quantization effects.
\newblock \emph{IEEE Transactions on Automatic Control}, 54\penalty0
  (11):\penalty0 2506--2517, 2009.

\bibitem[Nedi{\'c} et~al.(2017)Nedi{\'c}, Olshevsky, and
  Rabbat]{nedic2017network}
A.~Nedi{\'c}, A.~Olshevsky, and M.~G. Rabbat.
\newblock Network topology and communication-computation tradeoffs in
  decentralized optimization.
\newblock \emph{arXiv preprint arXiv:1709.08765}, 2017.

\bibitem[Nemirovski et~al.(2009)Nemirovski, Juditsky, Lan, and
  Shapiro]{Nemi_dsgd}
A.~Nemirovski, A.~Juditsky, G.~Lan, and A.~Shapiro.
\newblock Robust stochastic approximation approach to stochastic programming.
\newblock \emph{SIAM Journal on Optimization}, 19\penalty0 (4):\penalty0
  1574--1609, 2009.
\newblock \doi{10.1137/070704277}.

\bibitem[Omidshafiei et~al.(2017)Omidshafiei, Pazis, Amato, How, and
  Vian]{omidshafiei2017deep}
S.~Omidshafiei, J.~Pazis, C.~Amato, J.~P. How, and J.~Vian.
\newblock Deep decentralized multi-task multi-agent rl under partial
  observability.
\newblock \emph{arXiv preprint arXiv:1703.06182}, 2017.

\bibitem[Ram et~al.(2010{\natexlab{a}})Ram, Nedi{\'c}, and
  Veeravalli]{ram2010asynchronous}
S.~S. Ram, A.~Nedi{\'c}, and V.~V. Veeravalli.
\newblock Asynchronous gossip algorithm for stochastic optimization: Constant
  stepsize analysis.
\newblock In \emph{Recent Advances in Optimization and its Applications in
  Engineering}, pages 51--60. Springer, 2010{\natexlab{a}}.

\bibitem[Ram et~al.(2010{\natexlab{b}})Ram, Nedi{\'c}, and
  Veeravalli]{ram2010distributed}
S.~S. Ram, A.~Nedi{\'c}, and V.~V. Veeravalli.
\newblock Distributed stochastic subgradient projection algorithms for convex
  optimization.
\newblock \emph{Journal of optimization theory and applications}, 147\penalty0
  (3):\penalty0 516--545, 2010{\natexlab{b}}.

\bibitem[Schmidt et~al.(2017)Schmidt, Le~Roux, and Bach]{schmidt2017minimizing}
M.~Schmidt, N.~Le~Roux, and F.~Bach.
\newblock Minimizing finite sums with the stochastic average gradient.
\newblock \emph{Mathematical Programming}, 162\penalty0 (1-2):\penalty0
  83--112, 2017.

\bibitem[Seide and Agarwal(2016)]{Seide:2016:CMO:2939672.2945397}
F.~Seide and A.~Agarwal.
\newblock Cntk: Microsoft's open-source deep-learning toolkit.
\newblock In \emph{Proceedings of the 22Nd ACM SIGKDD International Conference
  on Knowledge Discovery and Data Mining}, KDD '16, pages 2135--2135, New York,
  NY, USA, 2016. ACM.
\newblock ISBN 978-1-4503-4232-2.
\newblock \doi{10.1145/2939672.2945397}.

\bibitem[Shahrampour and Jadbabaie(2017)]{shahrampour2017distributed}
S.~Shahrampour and A.~Jadbabaie.
\newblock Distributed online optimization in dynamic environments using mirror
  descent.
\newblock \emph{IEEE Transactions on Automatic Control}, 2017.

\bibitem[Shi et~al.(2014)Shi, Ling, Yuan, Wu, and Yin]{shi2014linear}
W.~Shi, Q.~Ling, K.~Yuan, G.~Wu, and W.~Yin.
\newblock On the linear convergence of the admm in decentralized consensus
  optimization.
\newblock \emph{IEEE Trans. Signal Processing}, 62\penalty0 (7):\penalty0
  1750--1761, 2014.

\bibitem[Shi et~al.(2015{\natexlab{a}})Shi, Ling, Wu, and Yin]{shi2015extra}
W.~Shi, Q.~Ling, G.~Wu, and W.~Yin.
\newblock Extra: An exact first-order algorithm for decentralized consensus
  optimization.
\newblock \emph{SIAM Journal on Optimization}, 25\penalty0 (2):\penalty0
  944--966, 2015{\natexlab{a}}.

\bibitem[Shi et~al.(2015{\natexlab{b}})Shi, Ling, Wu, and Yin]{shi2015proximal}
W.~Shi, Q.~Ling, G.~Wu, and W.~Yin.
\newblock A proximal gradient algorithm for decentralized composite
  optimization.
\newblock \emph{IEEE Transactions on Signal Processing}, 63\penalty0
  (22):\penalty0 6013--6023, 2015{\natexlab{b}}.

\bibitem[Suresh et~al.(2017)Suresh, Yu, Kumar, and McMahan]{pmlr-v70-suresh17a}
A.~T. Suresh, F.~X. Yu, S.~Kumar, and H.~B. McMahan.
\newblock Distributed mean estimation with limited communication.
\newblock In D.~Precup and Y.~W. Teh, editors, \emph{Proceedings of the 34th
  International Conference on Machine Learning}, volume~70 of \emph{Proceedings
  of Machine Learning Research}, pages 3329--3337, International Convention
  Centre, Sydney, Australia, 06--11 Aug 2017. PMLR.
\newblock URL \url{http://proceedings.mlr.press/v70/suresh17a.html}.

\bibitem[Yuan et~al.(2016)Yuan, Ling, and Yin]{Yuan_dsgd}
K.~Yuan, Q.~Ling, and W.~Yin.
\newblock On the convergence of decentralized gradient descent.
\newblock \emph{SIAM Journal on Optimization}, 26\penalty0 (3):\penalty0
  1835--1854, 2016.
\newblock \doi{10.1137/130943170}.

\bibitem[Yuan et~al.(2017)Yuan, Ying, Zhao, and Sayed]{yuan2017exact}
K.~Yuan, B.~Ying, X.~Zhao, and A.~H. Sayed.
\newblock Exact diffusion for distributed optimization and learning---part i:
  Algorithm development.
\newblock \emph{arXiv preprint arXiv:1702.05122}, 2017.

\bibitem[Zhang et~al.(2017)Zhang, Zhao, Zhu, Hoi, and
  Zhang]{zhang2017projection}
W.~Zhang, P.~Zhao, W.~Zhu, S.~C. Hoi, and T.~Zhang.
\newblock Projection-free distributed online learning in networks.
\newblock In \emph{International Conference on Machine Learning}, pages
  4054--4062, 2017.

\end{thebibliography}
%
%
%
%
%
%
%
%
\newpage
\onecolumn
\begin{center}
{\bf \Huge Supplemental Materials}
\end{center}

This supplement material includes the proofs for Theorem~\ref{theo_1}.

Because the confusion matrix $W$ is symmetric, it can be decomposed as $W = P\Lambda P^{\top}$, where $P = \left(\bm{v}_1,\bm{v}_2,\cdots,\bm{v}_n\right)$ is an orthogonal matrix, i.e., $P^{\top}P=PP^{\top}=I$, and $\Lambda=\textnormal{diag}\{\lambda_1,\dots,\lambda_n\}$ is a diagonal matrix with diagonal entries being the eigenvalues of $W$ in the nonincreasing order.
Then applying the decomposition to the iteration (from $W_{t}$ and $W_{t-1}$ to $W_{t+1}$)
\begin{align*}
X_{t+1} =  2X_tW-X_{t-1}W-\gamma G(X_t;\xi_t)W + \gamma G(X_{t-1};\xi_{t-1})W
\end{align*}
gives
\begin{align*}
X_{t+1} = & 2X_tP\Lambda P^{\top}-X_{t-1}P\Lambda P^{\top}-\gamma G(X_t;\xi_t)P\Lambda P^{\top} + \gamma G(X_{t-1};\xi_{t-1})P\Lambda P^{\top}.
\end{align*}
Denote $Y_t = X_tP$, $H(X_t;\xi_t) = G(X_t;\xi_t)P$, and use $y_t^{(i)}$ and $h_t^{(i)}$ to indicate the $i$-th column of $Y_t$ and $H(X_t;\xi_t)$, respectively.
Then
\begin{align*}
Y_{t+1} = & 2Y_t\Lambda - Y_{t-1}\Lambda - \gamma H(X_t;\xi_t)\Lambda + \gamma H(X_{t-1};\xi_{t-1})\Lambda,\label{Y_t}\numberthis
\end{align*}
or in the columns of $Y_t$ and $H(X_t;\xi_t)$,
\begin{align*}
y_{t+1}^{(i)} = & \lambda_i(2y_{t}^{(i)} - y_{t-1}^{(i)} - \gamma h_t^{(i)} + \gamma h_{t-1}^{(i)}).\label{y_t}\numberthis
\end{align*}

From the properties of $W$ in Assumption~\ref{ass:global} and the decomposition, we have $\lambda_1= 1$ and $\bm{v}_1 = \frac{1}{\sqrt{n}}(1,1,\cdots,1)^\top$. Therefore $y_t^{(1)} = \overline{X}_t\sqrt{n}$. For all other eigenvalues $-\frac{1}{3} < \lambda_i < 1$, the equation~\eqref{y_t} shows that all $y_t^{(i)}$ would ``decay to zero'', which explains how the confusion matrix works.

\begin{lemma}\label{lemma_rho_at}
Given two non-negative sequences $\{a_t\}_{t=1}^{\infty}$ and $\{b_t\}_{t=1}^{\infty}$ that satisfying
\begin{equation}
a_t =  \sum_{s=1}^t\rho^{t-s}b_{s}, \numberthis \label{eqn1}
\end{equation}
with $\rho\in[0,1)$, we have
\begin{align*}
S_k:=\sum_{t=1}^{k}a_t   \leq & \sum_{s=1}^k\frac{b_s}{1-\rho},\\
D_k:=\sum_{t=1}^{k}a_t^2 \leq & 
\frac{1}{(1-\rho)^2} \sum_{s=1}^kb_s^2.
\end{align*}
\end{lemma}

\begin{proof}
\begin{align}
S_k:=\sum_{t=1}^{k}a_t  = & \sum_{t=1}^{k}\sum_{s=1}^t\rho^{t-s}b_{s} = \sum_{s=1}^{k}\sum_{t=s}^k\rho^{t-s}b_{s} = \sum_{s=1}^{k}\sum_{t=0}^{k-s}\rho^{t}b_{s}\leq  \sum_{s=1}^{k}{b_{s}\over 1-\rho}. \label{eqn3}\\
D_k:=\sum_{t=1}^{k}a_t^2  = & \sum_{t=1}^{k}\sum_{s=1}^t\rho^{t-s}b_{s}\sum_{r=1}^t\rho^{t-r}b_{r} =  \sum_{t=1}^{k}\sum_{s=1}^t\sum_{r=1}^t\rho^{2t-s-r}b_{s}b_{r} \nonumber\\
\leq &  \sum_{t=1}^{k}\sum_{s=1}^t\sum_{r=1}^t\rho^{2t-s-r}{b_{s}^2+b_{r}^2\over2} =   \sum_{t=1}^{k}\sum_{s=1}^t\sum_{r=1}^t\rho^{2t-s-r}b_{s}^2 \nonumber\\
\leq  & {1\over 1-\rho}\sum_{t=1}^{k}\sum_{s=1}^t\rho^{t-s}b_{s}^2  \leq {1\over (1-\rho)^2}\sum_{s=1}^{k}b_{s}^2
\end{align}
where the last inequality holds because of~\eqref{eqn3}.
\end{proof}

\begin{lemma}\label{lemma_bound_transformation}
For any matrix $X_t\in \mathbb{R}^{N\times n}$,
we have
\begin{align*}
\sum_{i=2}^n\left\| X_t\bm{v}^{(i)} \right\|^2 \leq & \sum_{i=1}^n\left\| X_t\bm{v}^{(i)} \right\|^2 =
\left\| X_t\right\|^2_F\\
\sum_{i=1}^n\left\| X_tP^{\top}\bm{e}^{(i)} \right\|^2 = & \left\| X_tP^{\top}\right\|^2_F = \left\| X_t\right\|^2_F
\end{align*}
where $\bm{e}^{(i)}\in \mathbb{R}^{n\times 1}$ with the $i$-th component being 1 and all others being 0.
\end{lemma}
\begin{proof}
From the definition of the Frobenius norm for a matrix, we have
\begin{align*}
\sum_{i=1}^n\left\| X_t\bm{v}^{(i)} \right\|^2 = & \|X_tP\|_F^2 =  Tr\left(X_tPP^{\top}X_t^{\top}\right)=  Tr\left(X_tX_t^{\top}\right) =  \left\| X_t\right\|^2_F.
\end{align*}
Since $\left\| X_t\bm{v}^{(1)} \right\|^2\geq 0$, so
\begin{align*}
\sum_{i=2}^n\left\| X_t\bm{v}^{(i)} \right\|^2
\leq \sum_{i=1}^n\left\| X_t\bm{v}^{(i)} \right\|^2 = \sum_{i=1}^n\left\| X_t\right\|^2_F.
\end{align*}
In the same way, we have
\begin{align*}
\sum_{i=1}^n\left\| X_tP^{\top}\bm{e}^{(i)} \right\|^2 =  \left\| X_tP^{\top}\right\|^2_F = \left\| X_t\right\|^2_F.
\end{align*}
The result is proved.
\end{proof}

\begin{lemma}\label{lemma_a_and_2a}
Given $\rho \in (-{1\over3},0)\bigcup (0,1)$, for any two sequence $\{a_t\}_{t=0}^\infty$ and $\{b_t\}_{t=0}^\infty$ that satisfy
\begin{align*}
a_0 = & b_0 = 0, \\
a_1 = & b_1,\\
a_{t+1} = & \rho(2a_t -a_{t-1}) + b_t - b_{t-1}, \quad\forall  t\geq 1,
\end{align*}
we have
\begin{align*}
a_{t+1} = a_1\left(\frac{u^{t+1}-v^{t+1}}{u-v}\right) + \sum_{s=1}^t\beta_s\frac{u^{t-s+1}-v^{t-s+1}}{u-v},\quad \forall t\geq 0,
\end{align*}
where
\begin{align*}
\beta_s =  b_s - b_{s-1},\quad u =  \rho + \sqrt{\rho^2 - \rho},\quad v =  \rho - \sqrt{\rho^2 - \rho}.
\end{align*}
More specifically, if $0< \rho < 1$, we have
\begin{align*}
a_{t+1}\sin{\theta} = a_1\rho^{t/2}\sin{[(t+1)\theta]} + \sum_{s=1}^t\beta_s\rho^{(t-s)/2}\sin{[(t+1-s)\theta]},\quad\forall t\geq 0
\end{align*}
where
\begin{align*}
\beta_s = & b_s - b_{s-1},\quad \theta=\arccos(\sqrt{\rho}).
\end{align*}
\end{lemma}

\begin{proof}
When $t=0$, the results is easy to verify.
Next we consider the case $t\geq1$.
Since
\begin{align*}
a_{t+1} =  2\rho a_t - \rho a_{t-1} + \beta_t,
\end{align*}
We can find
\begin{align*}
u =  \rho + \sqrt{\rho^2 - \rho}, \quad  v =  \rho - \sqrt{\rho^2 - \rho},
\end{align*}
such that
\begin{align*}
a_{t+1} - ua_t = & (a_t - ua_{t-1})v + \beta_t.\numberthis \label{lemma_an_a-ua} \\
\end{align*}
Note that $u$ and $v$ are complex numbers when $0<\rho<1$. That is
\begin{align*}
u= \sqrt{\rho}e^{i\theta},\quad v= \sqrt{\rho}e^{-i\theta},
\end{align*}
with $\theta=\arccos{(\sqrt{\rho})}$.



Recursively applying~\eqref{lemma_an_a-ua} gives
\begin{align*}
a_{t+1} - ua_t = & (a_t - ua_{t-1})v + \beta_t =   (a_{t-1} - ua_{t-2})v^2 + \beta_{t-1}v + \beta_t \\
=  & (a_1 - ua_0)v^t + \sum_{s=1}^t\beta_{s}v^{t-s} \\
=  & a_1v^t + \sum_{s=1}^t\beta_{s}v^{t-s}. \quad\text{(due to $a_0=0$)}
\end{align*}
Diving both sides by $u^{t+1}$, we obtain
\begin{align*}
\frac{a_{t+1}}{u^{t+1}} = & \frac{a_t}{u^t}  + u^{-(t+1)}\left(a_1v^t + \sum_{s=1}^t\beta_sv^{t-s}\right) \\
= &\frac{a_{t-1}}{u^{t-1}}  + u^{-t}\left(a_1v^{t-1} + \sum_{s=1}^{t-1}\beta_sv^{t-1-s}\right) + u^{-(t+1)}\left(a_1v^t + \sum_{s=1}^t\beta_sv^{t-s}\right)\\
= & \frac{a_1}{u} + \sum_{k=1}^tu^{-k-1}\left(a_1v^k + \sum_{s=1}^k\beta_sv^{k-s}\right)
\end{align*}
Then we multiply both sides by $u^{t+1}$ and have
\begin{align*}
a_{t+1} = & a_1u^t + \sum_{k=1}^tu^{t-k}\left(a_1v^k + \sum_{s=1}^k\beta_sv^{k-s}\right)\\
= & a_1u^t\left(1 + \sum_{k=1}^t\left(\frac{v}{u}\right)^k\right) +  u^t\sum_{k=1}^t\sum_{s=1}^k\beta_sv^{-s}\left(\frac{v}{u}\right)^{k}\\
= & a_1u^t\sum_{k=0}^t\left(\frac{v}{u}\right)^k+  u^t\sum_{s=1}^t\sum_{k=s}^t\beta_sv^{-s}\left(\frac{v}{u}\right)^{k}\quad \text{(due to $\sum_{k=1}^t\sum_{s=1}^ka_sb_k = \sum_{s=1}^t\sum_{k=s}^ta_sb_k$)} \\
= & a_1u^t\left(\frac{1-\left(\frac{v}{u}\right)^{t+1}}{1-\frac{v}{u}}\right) +  u^t\sum_{s=1}^t\beta_sv^{-s}\left(\frac{v}{u}\right)^{s}\frac{1-\left(\frac{v}{u}\right)^{t-s+1}}{1-\frac{v}{u}}\\
= & a_1\left(\frac{u^{t+1}-v^{t+1}}{u-v}\right) + \sum_{s=1}^t\beta_s\frac{u^{t-s+1}-v^{t-s+1}}{u-v}.
\end{align*}
When $\rho \in (0,1)$,
since $u = \sqrt{\rho} e^{i\theta}$ and $v = \sqrt{\rho} e^{-i\theta}$, we have
\begin{align*}
a_{t+1} = a_1\rho^{t/2}\frac{\sin{[(t+1)\theta]}}{\sin{\theta}} + \sum_{s=1}^t\beta_s\rho^{(t-s)/2}\frac{\sin{[(t-s+1)\theta]}}{\sin{\theta}}.
\end{align*}
The result is proved.
\end{proof}

\begin{lemma} \label{lemma_bound_x_ave}
Under Assumption~\ref{ass:global}, we have
\begin{align*}
& \left(1-24C_2\gamma^2L^2\right)\sum_{i=1}^{n}\sum_{t=0}^{T}\left\|\overline{X}_t-\bm{x}_t^{(i)}\right\|^2\\
& \leq 2C_1\|X_1\|^2_F + 12C_2\gamma^2n\sigma^2T + 6C_2\gamma^4L^2\sigma^2T + 6C_2\gamma^4L^2n \sum_{t=1}^{T-1}\left\|\overline{\nabla f}(X_t)\right\|^2,
\end{align*}
where $\gamma$, $L$, $\sigma$, $\theta$, $C_1$ and $C_2$ are defined in Theorem~\ref{theo_1}.
\end{lemma}

\begin{proof}
To estimate the difference of the local models and the global mean model, we have
\begin{align*}
\sum_{i=1}^n\left\|\overline{X}_t-\bm{x}_t^{(i)}\right\|^2= &
\sum_{i=1}^n\left\| X_t\bm{e}^{(i)}-X_t\frac{\bm{1}_n}{n} \right\|^2 = \|X_t-X_t{\bm{1}_n\bm{1}_n^\top\over n}\|_F^2 = \|X_tPP^\top-X_t\bm{v}_1\bm{v}_1^\top\|_F^2\\
 = & \left\| X_tP\begin{pmatrix}
&0,&0,&0,&\cdots,&0\\
&0,&1,&0,&\cdots,&0\\
&0,&0,&1,&\cdots,&0\\
& \hdotsfor{5}\\
&0,&0,&0,&\cdots,&1
\end{pmatrix}\right\|^2_F\\
 =  &\sum_{i=2}^n \left\|y_t^{(i)}\right\|^2,\numberthis \label{lemma_bound_substract_mean_eq0}
\end{align*}
where $y_t^{(i)}$ is the $i$-th column of $X_tP$.
Note that we have, from~\eqref{y_t},
$$y_{t+1}^{(i)} =  \lambda_i(2y_{t}^{(i)} - y_{t-1}^{(i)} - \gamma h_t^{(i)} + \gamma h_{t-1}^{(i)})=\lambda_i(2y_{t}^{(i)} - y_{t-1}^{(i)}) + \lambda_i\beta_t^{(i)},$$
where $\beta_t^{(i)} = -\gamma h_{t}^{(i)}+\gamma h_{t-1}^{(i)}$.
For all $y^{(i)}$ that corresponding to $-\frac{1}{3} < \lambda_i<0$, Lemma~\ref{lemma_a_and_2a} shows
\begin{align*}
y_{t+1}^{(i)} = & y_{1}^{(i)}\left(\frac{u^{t+1}_i-v_i^{t+1}}{u_i-v_i}\right) + \lambda_i\sum_{s=1}^t\beta_s^{(i)}\frac{u_i^{t-s+1}-v_i^{t-s+1}}{u_i-v_i},
\end{align*}
where $u_i = \lambda_i + \sqrt{\lambda^2_i - \lambda_i}$ and $v_i = \lambda_i - \sqrt{\lambda^2_i - \lambda_i}$.
Therefore, we have
\begin{align*}
\left\| y_{t+1}^{(i)} \right\|^2 \leq & 2\left\| y_{1}^{(i)} \right\|^2\left(\frac{u_i^{t+1}-v_i^{t+1}}{u_i-v_i}\right)^2 + 2\lambda_i^2\left( \sum_{s=1}^t\left\|\beta_s^{(i)}\right\| \left|\frac{u_i^{t-s+1}-v_i^{t-s+1}}{u_i-v_i}\right| \right)^2.\numberthis \label{lemma_bound_x_ave_lambda_neg_1}
\end{align*}

For $\left|\frac{u_i^{n+1}-v_i^{n+1}}{u_i-v_i}\right|$, we have
\begin{align*}
\left|\frac{u_i^{n+1}-v_i^{n+1}}{u_i-v_i}\right| \leq  |v_i|^{n} \left|\frac{u_i\left(\frac{u_i}{v_i}\right)^{n} - v_i}{u_i-v_i}\right| \leq  |v_i|^n \quad \text{(due to $|u_i|<|v_i|$)}.
\end{align*}
Using~\eqref{lemma_bound_x_ave_lambda_neg_1}, we obtain
\begin{align*}
\left\| y_{t+1}^{(i)} \right\|^2 \leq & 2\left\| y_{1}^{(i)} \right\| |v_i|^{2t} +  2\lambda_i^2\left( \sum_{s=1}^t\left\|\beta_s^{(i)}\right\| |v_i|^{t-s}\right)^2.
\end{align*}
Summing from $t=0$ to $t=T-1$ gives
\begin{align*}
\sum_{t=0}^{T-1}\left\| y_{t+1}^{(i)} \right\|^2=\sum_{t=1}^T\left\| y_{t}^{(i)} \right\|^2 \leq 2\left\| y_{1}^{(i)} \right\|\sum_{t=0}^{T-1}|v_i|^{2t} + 2\lambda_i^2\sum_{t=1}^{T-1}\left( \sum_{s=1}^t\left\|\beta_s^{(i)}\right\| |v_i|^{t-s}\right)^2.
\end{align*}
Denote $a_t = \sum_{s=1}^t\left\|\beta_s^{(i)}\right\| |v_i|^{t-s}$, which has the same structure as the sequence in Lemma~\ref{lemma_rho_at}. Therefore, when $\lambda_i < 0$, we have
\begin{align*}
\sum_{t=1}^T\left\| y_{t}^{(i)} \right\|^2 \leq & \frac{2\left\| y_{1}^{(i)} \right\|}{1-|v_i|^2} + \frac{2\lambda_i^2}{(1-|v_i|)^2}\sum_{t=1}^{T-1}\left\|\beta_t^{(i)}\right\|^2\\
\leq & \frac{2\left\| y_{1}^{(i)} \right\|}{1-|v|^2} + \frac{2\lambda_n^2}{(1-|v|)^2}\sum_{t=1}^{T-1}\left\|\beta_t^{(i)}\right\|^2,\numberthis \label{lambda_neg}
\end{align*}
where $v = \lambda_n - \sqrt{\lambda^2_n - \lambda_n}$.

For all $y^{(i)}$ that satisfies $0\leq \lambda_i<1$, from \eqref{y_t} and Lemma~\ref{lemma_a_and_2a}, we have
\begin{align*}
y_{t+1}^{(i)} \sin{\theta_i} = y_1^{(i)}\lambda_i^{t/2}\sin{[(t+1)\theta_i]} + \lambda_i\sum_{s=1}^{t}\beta_{s}^{(i)}\lambda_i^{(t-s)/2}\sin{[(t+1-s)\theta_i]},
\end{align*}
where $\beta_{s}^{(i)} = -\gamma h_s^{(i)} + \gamma h_{s-1}^{(i)}$ and $\theta_i = \arccos(\sqrt{\lambda_i})$.

Then
\begin{align*}
\left\|y_{t+1}^{(i)}\right\|^2\sin^2\theta_i
\leq & 2\left\|y_1^{(i)}\right\|^2\lambda_i^{t}\sin^2[(t+1)\theta_i] + 2\lambda_i^2\left(\sum_{s=1}^{t}\left\|\beta_{s}^{(i)}\sin{[(t+1-s)\theta_i]}\right\|\lambda_i^{(t-s)/2}\right)^2\\
\leq & 2\left\|y_1^{(i)}\right\|^2\lambda_i^{t} + 2\lambda_i^2\left(\sum_{s=1}^{t}\|\beta_{s}^{(i)}\|\lambda_i^{(t-s)/2}\right)^2,
\end{align*}
Summing from $t=0$ to $T-1$ gives
\begin{align*}
\sum_{t=0}^{T-1}\left\|y_{t+1}^{(i)}\right\|^2\sin^2\theta_i =\sum_{t=1}^{T}\left\|y_{t}^{(i)}\right\|^2\sin^2\theta_i
\leq &   2\left\|y_1^{(i)}\right\|^2 \sum_{t=0}^{T-1} \lambda_i^{t} + 2\lambda_i^2 \sum_{t=1}^{T-1} \left(\sum_{s=1}^{t}\|\beta_{s}^{(i)}\|\lambda_i^{(t-s)/2}\right)^2 
\end{align*}

From Lemma~\ref{lemma_rho_at}, $\sum_{s=1}^{t}\|\beta_{s}^{(i)}\|\lambda_i^{(t-s)/2}$ has the same structure as the sequence in Lemma~\ref{lemma_rho_at}, so we have
\begin{align*}
\sum_{t=1}^{T}\left\|y_{t}^{(i)}\right\|^2\sin^2\theta_i
\leq & \frac{2\left\|y_1^{(i)}\right\|^2}{1-\lambda_i} + \frac{2\lambda_i^2}{(1-\sqrt{\lambda_i})^2}\sum_{t=1}^{T-1}\left\|\beta_t^{(i)}\right\|^2.
\end{align*}
Then $\sin^2\theta_i=1-\lambda_i$ gives
\begin{align*}
\sum_{t=1}^{T}\left\|y_{t}^{(i)}\right\|^2\leq & \frac{2\left\|y_1^{(i)}\right\|^2}{(1-\lambda_i)^2} + \frac{2\lambda_i^2}{(1-\sqrt{\lambda_i})^2(1-\lambda_i)}\sum_{t=1}^{T-1}\left\|\beta_t^{(i)}\right\|^2\\
\leq & \frac{2\left\|y_1^{(i)}\right\|^2}{(1-\lambda)^2} + \frac{2\lambda^2}{(1-\sqrt{\lambda})^2(1-\lambda)}\sum_{t=1}^{T-1}\left\|\beta_t^{(i)}\right\|^2.\numberthis \label{lambda_pos}
\end{align*}
Denote $C_1 = \max \left\{\frac{1}{1-|v|^2},\frac{1}{(1-\lambda)^2}\right\}$ and $C_2 = \max \left\{\frac{\lambda_n^2}{(1-|v|^2)},\frac{\lambda^2}{(1-\sqrt{\lambda})^2(1-\lambda)}\right\}$. From \eqref{lambda_neg} and \eqref{lambda_pos}, we have
\begin{align*}
\sum_{t=1}^{T}\left\|y_{t}^{(i)}\right\|^2\leq & 2C_1\left\|y_1^{(i)}\right\|^2 + 2C_2\sum_{t=1}^{T-1}\left\|\beta_t^{(i)}\right\|^2.\numberthis \label{lemma_bound_x_ave_eq1}
\end{align*}
We next bound $\beta_{t}^{(i)}$
\begin{align*}
  & \mathbb{E}\sum_{i=2}^{n}\|\beta_{t}^{(i)}\|^2 \\
= & \sum_{i=2}^{n}\gamma^2\mathbb{E}\|h_t^{(i)}-h_{t-1}^{(i)}\|^2\\
= & \gamma^2\sum_{i=2}^{n}\mathbb{E}\left\|G\left(X_t;\xi_t\right)P\bm{e}^{(i)} - G\left(X_{t-1};\xi_{t-1}\right)P\bm{e}^{(i)}\right\|^2\\
\leq & \gamma^2\sum_{i=1}^{n}\mathbb{E}\left\|G\left(X_t;\xi_t\right)P\bm{e}^{(i)} - G\left(X_{t-1};\xi_{t-1}\right)P\bm{e}^{(i)}\right\|^2\\
= & \gamma^2\mathbb{E}\left\|G\left(X_t;\xi_t\right)P - G\left(X_{t-1};\xi_{t-1}\right)P\right\|^2_F\\
= & \gamma^2 \mathbb{E}\left\|G\left(X_t;\xi_t\right) - G\left(X_{t-1};\xi_{t-1}\right)\right\|^2_F \quad \text{(due to Lemma~\ref{lemma_bound_transformation} )} \\
= & \gamma^2\sum_{i=1}^n\mathbb{E}\left\|\nabla F_i\left(\bm{x}_t^{(i)};\xi^{(i)}_t\right)-\nabla F_i\left(\bm{x}_{t-1}^{(i)};\xi^{(i)}_{t-1}\right)\right\|^2\\
= & \gamma^2\sum_{i=1}^n\mathbb{E}\left\|\left(\nabla F_i\left(\bm{x}_t^{(i)};\xi^{(i)}_t\right)-\nabla f_i(\bm{x}_t^{(i)})\right)-\left( F_i\left(\bm{x}_{t-1}^{(i)};\xi^{(i)}_{t-1}\right) -\nabla f_i(\bm{x}_{t-1}^{(i)})\right) + \left( \nabla f_i\left(\bm{x}_t^{(i)}\right)-\nabla f_i\left(\bm{x}_{t-1}^{(i)}\right) \right) \right\|^2\\
= & 3\gamma^2\sum_{i=1}^n\mathbb{E}\left\|\nabla F_i\left(\bm{x}_t^{(i)};\xi^{(i)}_t\right)-\nabla f_i(\bm{x}_t^{(i)})\right\|^2 + 3\gamma^2\sum_{i=1}^n\left\|F_i\left(\bm{x}_{t-1}^{(i)};\xi^{(i)}_{t-1}\right) -\nabla f_i(\bm{x}_{t-1}^{(i)})\right\|^2\\
&  + 3\gamma^2\sum_{i=1}^n\left\|\nabla f_i\left(\bm{x}_t^{(i)}\right)-\nabla f_i\left(\bm{x}_{t-1}^{(i)}\right)\right\|^2\\
\leq & 6\gamma^2n\sigma^2 + 3\gamma^2\sum_{i=1}^n\mathbb{E}\left\|\nabla f_i\left(\bm{x}_t^{(i)}\right) - \nabla f_i\left(\bm{x}_{t-1}^{(i)}\right)\right\|^2\\
\leq & 6\gamma^2n\sigma^2 + 3\gamma^2\sum_{i=1}^nL^2\mathbb{E}\left\|\bm{x}_t^{(i)} -\bm{x}_{t-1}^{(i)}\right\|^2\\
= & 6\gamma^2n\sigma^2 + 3\gamma^2L^2\sum_{i=1}^n\mathbb{E}\left\|Y_{t}P^{\top}\bm{e}^{(i)}  - Y_{t-1}P^{\top}\bm{e}^{(i)}\right\|^2\\
= & 6\gamma^2n\sigma^2 + 3\gamma^2L^2\mathbb{E}\left\|Y_{t}P^{\top} - Y_{t-1}P^{\top}\right\|^2_F\\
= & 6\gamma^2n\sigma^2 + 3\gamma^2L^2\mathbb{E}\left\|Y_{t} - Y_{t-1}\right\|^2_F \quad \text{(due to Lemma~\ref{lemma_bound_transformation} )} \\
= & 6\gamma^2n\sigma^2 + 3\gamma^2L^2\sum_{i=1}^n\mathbb{E}\left\|y_{t}^{(i)} - y_{t-1}^{(i)}\right\|^2. \numberthis \label{lemma_bound_x_ave_bound_beta_l}
\end{align*}
Combing \eqref{lemma_bound_x_ave_eq1} and \eqref{lemma_bound_x_ave_bound_beta_l}, we have
\begin{align*}
\sum_{i=2}^{n}\sum_{t=1}^{T}\|y_{t}^{(i)}\|^2
\leq & 2C_1 \|Y_1\|^2_F + 2C_2 \sum_{i=2}^{n}\sum_{t=1}^{T-1}\|\beta_t^{(i)}\|^2\\
\leq & 2C_1\|Y_1\|^2_F + 2C_2\sum_{t=1}^{T-1}\left(6\gamma^2n\sigma^2 + 3\gamma^2L^2\sum_{i=1}^n\mathbb{E}\left\|y_{t}^{(i)} - y_{t-1}^{(i)}\right\|^2\right)\\
\leq & 2C_1\|Y_1\|^2_F + 12C_2\gamma^2n\sigma^2T + 6C_2\gamma^2L^2 \sum_{i=1}^n\sum_{t=1}^{T-1}\mathbb{E}\left\|y_{t}^{(i)} - y_{t-1}^{(i)}\right\|^2.\numberthis \label{lemma_bound_x_ave_eq2}
\end{align*}
The next step is to bound $\mathbb{E}\|y_{t}^{(1)} - y_{t-1}^{(1)}\|^2$. Because
\begin{align*}
y_t^{(1)} =  X_tP\bm{e}^{(1)} =  X_t\bm{v}_1=   X_t\frac{1}{\sqrt{n}}\bm{1}_n=  \overline{X_t}\sqrt{n},
\end{align*}
what we need to bound is $\mathbb{E}\|\overline{X}_{t+1}-\overline{X}_{t}\|^2$. From~\eqref{global_mean_evolution}, we have $\overline{X}_{t+1} =  \overline{X}_{t} - \gamma\overline{G}_{t}$. Therefore
\begin{align*}
\mathbb{E}\left\|\overline{X}_{t+1}-\overline{X}_{t}\right\|^2 =  \gamma^2\mathbb{E}\left\|\overline{G}_{t}\right\|^2
=\gamma^2 \mathbb{E}\|\overline{G}_{t}-\overline{\nabla}f(X_t)\|^2+ \gamma^2\left\|\overline{\nabla}f(X_t)\right\|^2\leq  \frac{\gamma^2\sigma^2}{n} + \gamma^2\left\|\overline{\nabla}f(X_t)\right\|^2,
\end{align*}
and we have the follow bound for $\mathbb{E}\|y_{t}^{(1)} - y_{t-1}^{(1)}\|^2$:
\begin{align*}
\mathbb{E}\left\|y_{t+1}^{(1)}-y_t^{(1)}\right\|^2
\leq & \gamma^2\sigma^2 + n\gamma^2\left\|\overline{\nabla}f(X_t)\right\|^2.\numberthis \label{lemma_bound_x_ave_y_t}
\end{align*}
Combing \eqref{lemma_bound_x_ave_eq2} and \eqref{lemma_bound_x_ave_y_t} we get
\begin{align*}
\sum_{i=2}^{n}\sum_{t=1}^{T}\|y_{t}^{(i)}\|^2
\leq & 2C_1\|Y_1\|^2_F + 12C_2\gamma^2n\sigma^2T + 6C_2\gamma^4L^2\sigma^2T + 6C_2\gamma^4L^2n \sum_{t=1}^{T-1}\left\|\overline{\nabla f}(X_t)\right\|^2\\
& + 6C_2\gamma^2L^2\sum_{i=2}^n\sum_{t=1}^{T-1}\mathbb{E}\left\|y_{t}^{(i)} - y_{t-1}^{(i)}\right\|^2\\
\leq & 2C_1\|Y_1\|^2_F + 12C_2\gamma^2n\sigma^2T + 6C_2\gamma^4L^2\sigma^2T + 6C_2\gamma^4L^2n \sum_{t=1}^{T-1}\left\|\overline{\nabla f}(X_t)\right\|^2 \\
& + 6C_2\gamma^2L^2\sum_{i=2}^n\sum_{t=1}^{T-1}2\mathbb{E}\left(\left\|y_{t}^{(i)}\right\|^2 + \left\|y_{t-1}^{(i)}\right\|^2\right)\\
\leq & 2C_1\|Y_1\|^2_F + 12C_2\gamma^2n\sigma^2T + 6C_2\gamma^4L^2\sigma^2T + 6C_2\gamma^4L^2n \sum_{t=1}^{T-1}\left\|\overline{\nabla f}(X_t)\right\|^2 \\
& + 6C_2\gamma^2L^2\sum_{i=2}^n\sum_{t=1}^{T-1}2\mathbb{E}\left(\left\|y_{t}^{(i)}\right\|^2 + \left\|y_{t}^{(i)}\right\|^2\right)\quad \text{(due to $y_{0}^{(i)}=\bm{0}$)}\\
\leq & 2C_1\|Y_1\|^2_F + 12C_2\gamma^2n\sigma^2T + 6C_2\gamma^4L^2\sigma^2T + 6C_2\gamma^4L^2n \sum_{t=1}^{T-1}\left\|\overline{\nabla f}(X_t)\right\|^2\\
& + 24C_2\gamma^2L^2\sum_{i=2}^n\sum_{t=1}^{T-1}\mathbb{E}\left\|y_{t}^{(i)}\right\|^2,\\
\left(1-24C_2\gamma^2L^2\right)\sum_{i=2}^{n}\sum_{t=1}^{T}\|y_{t}^{(i)}\|^2
\leq & 2C_1\|Y_1\|^2_F + 12C_2\gamma^2n\sigma^2T + 6C_2\gamma^4L^2\sigma^2T + 6C_2\gamma^4L^2n \sum_{t=1}^{T-1}\left\|\overline{\nabla f}(X_t)\right\|^2.
\end{align*}
Together with~\eqref{lemma_bound_substract_mean_eq0} and $X_0 = 0$, we have
\begin{align*}
\left(1-24C_2\gamma^2L^2\right)\sum_{i=1}^{n}\sum_{t=1}^{T}\left\|\overline{X}_t-\bm{x}_t^{(i)}\right\|^2
\leq & 2C_1\|Y_1\|^2_F + 12C_2\gamma^2n\sigma^2T + 6C_2\gamma^4L^2\sigma^2T + 6C_2\gamma^4L^2n \sum_{t=1}^{T-1}\left\|\overline{\nabla f}(X_t)\right\|^2\\
\text{(due to $\|X_1\|_F = \|Y_1\|_F$)}\quad\leq & 2C_1\|X_1\|^2_F + 12C_2\gamma^2n\sigma^2T + 6C_2\gamma^4L^2\sigma^2T + 6C_2\gamma^4L^2n \sum_{t=1}^{T-1}\left\|\overline{\nabla f}(X_t)\right\|^2.
\end{align*}
Actually, when $\lambda_n\leq -\frac{1}{3}$, we have $|v_n|\geq 1$, then $\|y^{(n)}_t\|^2\varpropto t $ and
\begin{align*}
\frac{1}{T}\sum_{i=1}^{n}\sum_{t=1}^{T}\left\|\overline{X}_t-\bm{x}_t^{(i)}\right\|^2 \leq T.
\end{align*}
The algorithm would fail to converge in this situation, and this is why $-1/3$ is the infimum of $\lambda_n$.
\end{proof}

\begin{lemma}\label{lemma_boundfplus}
Following the Assumption~\ref{ass:global}, we have
\begin{align*}
\mathbb{E}f(\overline{X}_{t+1})
\leq & \mathbb{E}f(\overline{X}_t)  -  \frac{\gamma_t}{2}\mathbb{E}\|\nabla f(\overline{X}_t)\|^2  -  \left(\frac{\gamma_t}{2}-\frac{L\gamma_t^2}{2}\right)\mathbb{E}\|\overline{\nabla f}(X_t)\|^2  +
	\frac{\gamma_t}{2}\mathbb{E}\|\nabla f(\overline{X}_t)  -
	\overline{\nabla f}(X_t)\|^2   +   \frac{L\gamma_t^2}{2n}\sigma^2.
\end{align*}
\end{lemma}
\begin{proof}
From \eqref{global_mean_evolution}, we have
\begin{align*}
\overline{X}_{t+1} =  \oX_t - \gamma_t \oG(X_t; \xi_t).
\end{align*}
From item 1 of Assumption~\ref{ass:global}, we know that $f$ has a $L$-Lipschitz continuous gradient. So, we have
\begin{align*}
\mathbb{E}f(\overline{X}_{t+1}) \le & \mathbb{E}f(\overline{X}_t)+\mathbb{E}\left\langle\nabla f(\overline{X}_t), -\gamma_t\overline{G}(X_t;\xi_t)\right\rangle  +  \frac{L}{2}\mathbb{E}\left\|-\gamma_t\overline{G}(X_t;\xi_t)\right\|^2 \\
= & \mathbb{E}f(\overline{X}_t)  +  \mathbb{E}\langle\nabla f(\overline{X}_t), -\gamma_t\mathbb{E}_{\xi_t}\overline{G}(X_t;\xi_t)\rangle  +  \frac{L\gamma_t^2}{2}\mathbb{E}\|\overline{G}(X_t;\xi_t)\|^2 \\
= &  \mathbb{E}f(\overline{X}_t)  -  \gamma_t\mathbb{E}\langle\nabla f(\overline{X}_t), \overline{\nabla f}	(X_t)\rangle  +   \frac{L\gamma_t^2}{2}\mathbb{E}\|(\overline{G}(X_t;\xi_t) - \overline{\nabla f}(X_t))+\overline{\nabla f}(X_t)\|^2\\
= & \mathbb{E}f(\overline{X}_t)  -  \gamma_t\mathbb{E}\langle\nabla f(\overline{X}_t), \overline{\nabla f}	(X_t)\rangle  +   \frac{L\gamma_t^2}{2}\mathbb{E}\|\overline{G}(X_t;\xi_t) - \overline{\nabla f}(X_t)\|^2  +  \frac{L\gamma_t^2}{2}\mathbb{E}\|\overline{\nabla f}(X_t)\|^2 \\
& +  {L\gamma_t^2}\mathbb{E}\langle\mathbb{E}_{\xi_t}\overline{G}(X_t;\xi_t) - \overline{\nabla f}(X_t),\overline{\nabla f}(X_t)\rangle\\
= & \mathbb{E}f(\overline{X}_t)  -  \gamma_t\mathbb{E}\langle\nabla f(\overline{X}_t), \overline{\nabla f}	(X_t)\rangle  +   \frac{L\gamma_t^2}{2}\mathbb{E}\|\overline{G}(X_t;\xi_t) - \overline{\nabla f}	(X_t)\|^2  +  \frac{L\gamma_t^2}{2}\mathbb{E}\|\overline{\nabla f}(X_t)\|^2\\
    = & \mathbb{E}f(\overline{X}_t)  -  \gamma_t\mathbb{E}\langle\nabla f(\overline{X}_t), \overline{\nabla f}	(X_t)\rangle  +   \frac{L\gamma_t^2}{2n^2}\mathbb{E}\left\|\sum_{i=1}^n\left(\nabla F_i(x_t^{(i)};\xi_t^{(i)}) - \nabla f_i	(x_t^{(i)})\right)\right\|^2\\
    & +  \frac{L\gamma_t^2}{2}\mathbb{E}\|\overline{\nabla f}(X_t)\|^2\\
    = & \mathbb{E}f(\overline{X}_t)  -  \gamma_t\mathbb{E}\langle\nabla f(\overline{X}_t), \overline{\nabla f}	(X_t)\rangle  +   \frac{L\gamma_t^2}{2n^2}\sum_{i=1}^n\mathbb{E}\left\|\nabla F_i(x_t^{(i)};\xi_t^{(i)}) - \nabla f_i	(x_t^{(i)})\right\|^2\\
    & + \sum_{i\neq i'}^n\mathbb{E}\left\langle \mathbb{E}_{\xi_t}\nabla F_i(x_t^{(i)};\xi_t^{(i)}) - \nabla f_i(x_t^{(i)}),\nabla \mathbb{E}_{\xi_t}F_{i'}(x_t^{(i')};\xi_t^{(i')}) - \nabla f_{i'}	(x_t^{(i')}) \right\rangle +  \frac{L\gamma_t^2}{2}\mathbb{E}\|\overline{\nabla f}(X_t)\|^2\\
\leq & \mathbb{E}f(\overline{X}_t)  -  \gamma_t\mathbb{E}\langle\nabla f(\overline{X}_t),
	\overline{\nabla f}(X_t)\rangle  +    \frac{L\gamma_t^2}{2n}\sigma^2 + \frac{L\gamma_t^2}{2}\mathbb{E}\|\overline{\nabla f}(X_t)\|^2\\
= & \mathbb{E}f(\overline{X}_t)  -  \frac{\gamma_t}{2}\mathbb{E}\|\nabla f(\overline{X}_t)\|^2  -  \frac{\gamma_t}{2}\mathbb{E}\|\overline{\nabla f}(X_t)\|^2  +
	\frac{\gamma_t}{2}\mathbb{E}\|\nabla f(\overline{X}_t)  -
	\overline{\nabla f}(X_t)\|^2  +  \frac{L\gamma_t^2}{2}\mathbb{E}\|\overline{\nabla f}(X_t)\|^2\\
		&  +   \frac{L\gamma_t^2}{2n}\sigma^2\quad \text{(due to $2\langle \bm{a},\bm{b}\rangle=\|\bm{a}\|^2+\|\bm{b}\|^2-\|\bm{a}-\bm{b}\|^2$)}\\
= &  \mathbb{E}f(\overline{X}_t)  -  \frac{\gamma_t}{2}\mathbb{E}\|\nabla f(\overline{X}_t)\|^2  -  (\frac{\gamma_t}{2}-\frac{L\gamma_t^2}{2})\mathbb{E}\|\overline{\nabla f}(X_t)\|^2  +
	\frac{\gamma_t}{2}\mathbb{E}\|\nabla f(\overline{X}_t)  -
	\overline{\nabla f}(X_t)\|^2    +   \frac{L\gamma_t^2}{2n}\sigma^2, \numberthis \label{lemma_boundfplus_eq}
\end{align*}
which completes the proof.
\end{proof}

{\bf \noindent Proof to Theorem \ref{theo_1}}
\begin{proof}
We first estimate the upper bound for
$\mathbb{E}\|\nabla f(\overline{X}_t)  -  \overline{\nabla f}(X_t)\|^2$:
\begin{align*}
\mathbb{E}\|\nabla f(\overline{X}_t)  -  \overline{\nabla f}(X_t)\|^2 = & \frac{1}{n^2}\mathbb{E}{\left\|
	\sum_{i=1}^n\left(\nabla f_i(\overline{X}_t)  -  \nabla f_i(\bm{x}_t^{(i)})\right)\right\|^2}\\
\leq & \frac{1}{n}\sum_{i=1}^n\mathbb{E}\left\|\nabla f_i(\overline{X}_t)  -  \nabla f_i(\bm{x}_t^{(i)})\right\|^2\\
\leq & \frac{L^2}{n}\mathbb{E}\sum_{i=1}^n\left\|\overline{X}_t-\bm{x}_t^{(i)}\right\|^2. \numberthis \label{theo:boundnabla}
\end{align*}

Combining \eqref{lemma_boundfplus_eq} in Lemma~\ref{lemma_boundfplus} and \eqref{theo:boundnabla} yields
\begin{align*}
& \frac{\gamma_t}{2}\mathbb{E}\|\nabla f(\overline{X_t})\|^2
	+ \left(\frac{\gamma_t}{2}-\frac{L\gamma_t^2}{2}\right)\mathbb{E}\|\overline{\nabla f}(X_t)\|^2
	\\
\leq & \mathbb{E}f(\overline{X}_t)-\mathbb{E}f(\overline{X}_{t+1})
	+ \frac{\gamma_t}{2}\mathbb{E}\|\nabla f(\overline{X_t})  - \overline{\nabla f}(X_t)\|^2
 	 + \frac{L\gamma_t^2}{2n}\sigma^2\\
\leq &  \mathbb{E}f(\overline{X}_t)-\mathbb{E}f(\overline{X}_{t+1})
	+ \frac{L^2\gamma_t}{2n}\sum_{i=1}^{n}\|\overline{X}_t-x_t^{(i)}\|^2
	+  \frac{L\gamma_t^2}{2n}\sigma^2.
\end{align*}

Setting $\gamma_t=\gamma$,  we obtain
\begin{align*}
\mathbb{E}\|\nabla f(\overline{X_t})\|^2 + \left(1-L\gamma\right)\mathbb{E}\|\overline{\nabla f}(X_t)\|^2
\leq  \frac{2}{\gamma}\left(\mathbb{E}f(\overline{X}_t)-f^*-\left(\mathbb{E}f(\overline{X}_{t+1})-f^*\right)\right)
	+ \frac{L^2}{n}\sum_{i=1}^{n}\|\overline{X}_t-x_t^{(i)}\|^2 +  \frac{L\gamma}{n}\sigma^2. \numberthis \label{theo_eq_1}
\end{align*}
From Lemma~\ref{lemma_bound_x_ave} , we have
\begin{align*}
\left(1-24C_2\gamma^2L^2\right)\sum_{i=1}^{n}\sum_{t=0}^{T}\left\|\overline{X}_t-\bm{x}_t^{(i)}\right\|^2
& \leq 2C_1\|X_1\|^2_F + 12C_2\gamma^2n\sigma^2T + 6C_2\gamma^4L^2\sigma^2T\\
& + 6C_2\gamma^4L^2n \sum_{t=1}^{T-1}\left\|\overline{\nabla f}(X_t)\right\|^2,
\end{align*}
If $\gamma$ is not too large that satisfies $1-24C_2\gamma^2L^2 > 0$, then denote $C_3 = 1-24C_2\gamma^2L^2$, we would have
\begin{align*}
\sum_{i=1}^{n}\sum_{t=0}^{T}\left\|\overline{X}_t-\bm{x}_t^{(i)}\right\|^2
\leq & \frac{2C_1}{C_3}\|X_1\|^2_F + \frac{12C_2\gamma^2n\sigma^2T}{C_3} + \frac{6C_2\gamma^4L^2\sigma^2T}{C_3} + \frac{6C_2\gamma^4L^2n}{C_3} \sum_{t=1}^{T-1}\left\|\overline{\nabla f}(X_t)\right\|^2. \numberthis \label{them_bound_x_ave}
\end{align*}
 Summarizing both sides of \eqref{theo_eq_1} and applying \eqref{them_bound_x_ave} yields
\begin{align*}
& \sum_{t=0}^{T-1}\left(\mathbb{E}\|\nabla f(\overline{X_t})\|^2 + \left(1-L\gamma\right)\mathbb{E}\|\overline{\nabla f}(X_t)\|^2\right) \\
\leq & \frac{2\mathbb{E}f(\overline{X}_0)-2f^*}{\gamma} +\frac{L^2}{n}\sum_{t=0}^{T}\sum_{i=1}^n\mathbb{E}\left\|\overline{X}_t-x_t^{(i)}\right\|^2 + \frac{LT\gamma}{n}\sigma^2  \\
\leq & \frac{2(f(0)-f^*)}{\gamma}
 + \frac{LT\gamma}{n}\sigma^2 + \frac{2L^2C_1}{nC_3}\|X_1\|^2_F + \frac{12L^2C_2\gamma^2\sigma^2T}{C_3} + \frac{6L^2C_2\gamma^4L^2\sigma^2T}{nC_3}\\
 & + \frac{6L^2C_2\gamma^4L^2}{C_3} \sum_{t=1}^{T-1}\left\|\overline{\nabla f}(X_t)\right\|^2.
\end{align*}
It implies
\begin{align*}
& \sum_{t=0}^{T-1}\left(\mathbb{E}\|\nabla f(\overline{X_t})\|^2 + \left(1-L\gamma-\frac{6L^2C_2\gamma^4L^2}{C_3}\right)\mathbb{E}\|\overline{\nabla f}(X_t)\|^2\right)
\\
\leq & \frac{2(f(0)-f^*)}{\gamma}
 + \frac{LT\gamma}{n}\sigma^2 + \frac{2L^2C_1}{nC_3}\|X_1\|^2_F + \frac{12L^2C_2\gamma^2n\sigma^2T}{nC_3} + \frac{6L^2C_2\gamma^4L^2\sigma^2T}{nC_3}\\
= & \frac{2(f(0)-f^*)}{\gamma}
 + \frac{LT\gamma}{n}\sigma^2 + \frac{2L^2C_1\gamma^2}{nC_3}\|G(0;\xi_0)\|^2_F + \frac{12L^2C_2\gamma^2\sigma^2T}{C_3} + \frac{6L^2C_2\gamma^4L^2\sigma^2T}{nC_3}.
\numberthis \label{theo_final_1}
\end{align*}
However, $\|G(0;\xi_0)\|^2_F$ can be expanded as:
\begin{align*}
\|G(0,\xi_0)\|^2_F = & \sum_{i=1}^n\left\|\left(\nabla F_i(\bm{0},\xi_1) - \nabla f_i(\bm{0})\right) + \left(\nabla f_i(\bm{0}) - \nabla f(\bm{0})\right) + \nabla f(\bm{0})\right\|^2\\
\leq & 3n\sigma^2 + 3n\zeta_0^2 + 3n\|\nabla f(\bm{0})\|^2, \numberthis \label{bound_G_0}
\end{align*}
where $\zeta_0 = \frac{1}{n}\sum_{i=1}^n \|\nabla f_i(\bm{0}) - \nabla f(\bm{0})\|^2$ indicates the difference between different workers' dataset at the start point.
Combining \eqref{theo_final_1} and \eqref{bound_G_0}, then we have
\begin{align*}
& \sum_{t=0}^{T-1}\left(\mathbb{E}\|\nabla f(\overline{X_t})\|^2 + \left(1-L\gamma-\frac{6L^2C_2\gamma^4L^2}{C_3}\right)\mathbb{E}\|\overline{\nabla f}(X_t)\|^2\right)\\
\leq & \frac{2(f(0)-f^*)}{\gamma}
 + \frac{LT\gamma}{n}\sigma^2  + \frac{12L^2C_2\gamma^2\sigma^2T}{C_3} + \frac{6L^2C_2\gamma^4L^2\sigma^2T}{nC_3}\\
 & + \frac{6L^2C_1\gamma^2\sigma^2}{C_3} + \frac{6L^2C_1\gamma^2\zeta_0^2}{C_3} + \frac{6L^2C_1\gamma^2}{C_3}\|\nabla f(\bm{0})\|^2.
\end{align*}
Then we have
\begin{align*}
& \left(1 - \frac{6L^2C_1\gamma^2}{C_3}\right)\|\nabla f(\bm{0})\|^2 + \sum_{t=1}^{T-1}\left(\mathbb{E}\|\nabla f(\overline{X_t})\|^2 + \left(1-L\gamma-\frac{6L^2C_2\gamma^4L^2}{C_3}\right)\mathbb{E}\|\overline{\nabla f}(X_t)\|^2\right)\\
\leq & \frac{2(f(0)-f^*)}{\gamma}
 + \frac{LT\gamma}{n}\sigma^2  + \frac{12L^2C_2\gamma^2\sigma^2T}{C_3} + \frac{6L^2C_2\gamma^4L^2\sigma^2T}{nC_3}
  + \frac{6L^2C_1\gamma^2\sigma^2}{C_3} + \frac{6L^2C_1\gamma^2\zeta_0^2}{C_3}.
\end{align*}
Denote
\begin{align*}
A_1 = & 1 - \frac{6L^2C_1\gamma^2}{C_3}\\
A_2 = & 1-L\gamma-\frac{6L^2C_2\gamma^4L^2}{C_3},
\end{align*}
it becomes
\begin{align*}
& A_1\|\nabla f(\bm{0})\|^2 + \sum_{t=1}^{T-1}\left(\mathbb{E}\|\nabla f(\overline{X_t})\|^2 + A_2\mathbb{E}\|\overline{\nabla f}(X_t)\|^2\right)\\
\leq & \frac{2(f(0)-f^*)}{\gamma}
 + \frac{LT\gamma}{n}\sigma^2  + \frac{12L^2C_2\gamma^2n\sigma^2T}{nC_3} + \frac{6L^2C_2\gamma^4L^2\sigma^2T}{nC_3}
  + \frac{6L^2C_1\gamma^2\sigma^2}{C_3} + \frac{6L^2C_1\gamma^2\zeta_0^2}{C_3}.
\end{align*}
It completes the proof.
\end{proof}

{\bf \noindent Proof to \Cref{corollary2}}
\begin{proof}
From the value of $\gamma$, we obtain
\begin{align*}
C_2\gamma^2L^2 \leq \frac{1}{64},\quad  C_1\gamma^2L^2 \leq \frac{1}{36}.
\end{align*}
Therefore
\begin{align*}
C_3 = &1-24C_2\gamma^2L^2 \geq \frac{1}{2},\\
A_1 = &1-\frac{6L^2C_1\gamma^2}{C_3} \geq \frac{1}{2},\\
A_2 = & 1-L\gamma -\frac{6L^2C_2\gamma^4L^2}{C_3} > 0,\\
\gamma^2 \leq & \frac{n}{nL^2 + \sigma^2T},\\
\gamma^4 \leq & \frac{n^2}{n^2L^4 + \sigma^4T^2}.
\end{align*}
Then we can remove the $\|\overline{\nabla f}(X_t)\|^2$ and $\|\nabla f(\bm{0})\|^2 $ on the left hand side of~\eqref{eq:thm2} in Theorem~\ref{theo_1}, and~\eqref{eq:thm2} becomes
\begin{align*}
  \frac{1}{T}\sum_{t=0}^{T-1}\mathbb{E}\|\nabla f(\overline{X_t})\|^2
\leq & \frac{4(f(0) - f^*)L(8\sqrt{C_2} + 6\sqrt{C_1} )}{T} + \frac{4(f(0) - f^*)\sigma}{\sqrt{Tn}} \\
&+ \frac{2L\sigma}{\sqrt{Tn}} + \frac{48nL^2C_2\sigma^2}{nL^2 + \sigma^2T} + \frac{24L^4n\sigma^2C_2}{n^2L^4 + \sigma^4T^2} \\
&+ \frac{24nL^2C_1\sigma^2}{T(nL^2 + \sigma^2T)} + \frac{24L^2C_1\zeta_0^2}{T(nL^2 + \sigma^2T)},
\end{align*}
which completes the proof.
\end{proof}



\end{document}